\documentclass[11pt,a4paper]{article}
\usepackage[lmargin=1.1in, tmargin=0.8in, bmargin=0.9in]{geometry}
\usepackage[section]{placeins}

\usepackage[usenames,dvipsnames]{color}
\usepackage{pagecolor,lipsum}% http://ctan.org/pkg/{pagecolor,lipsum}

\newcommand{\fredN}[1]{ \marginpar{{ FredN: #1}} }

\newif\ifversionpublique
 \versionpubliquetrue % a commenter sinon 	
\ifversionpublique
	\renewcommand{\fredN}[1]{  }
\else
	\usepackage{showlabels}
	\pagecolor{yellow!10}
	\usepackage{showlabels}
\fi

\usepackage[utf8]{inputenc}

\usepackage{amsmath}
\usepackage{amsfonts}
\usepackage{amssymb}
\usepackage{mathabx}
\usepackage{bm}%bold math
\usepackage{makeidx}
\usepackage{hyperref}
\usepackage{verbatim}

\usepackage{psfrag}
\usepackage{graphicx}

\usepackage{algorithmic}
\usepackage[ruled]{algorithm}

\usepackage{bm}
\usepackage{stmaryrd}

\usepackage{csvsimple}
\usepackage{multirow}
\usepackage{siunitx}

\newtheorem{lemma}{Lemma}[section]

\newtheorem{assumption}{Assumption}[section]

\newtheorem{remark}{Remark}[section]

\newcommand{\N}{\mathbb N}

\newcommand{\R}{\mathbb R}

\newcommand{\HO}{H}
\newcommand{\HP}{H_D}
\newcommand{\RL}{{\mathcal R}}
\newcommand{\NC}{{k_0}}%NC : neighbours count
\newcommand{\id}{{I}}

\newcommand{\QED}{\hspace*{\fill}\rule{2.5mm}{2.5mm}}  
\newenvironment{proof}{{\bf Proof\ }}{\QED\\}

\newcommand{\newtablefromcsv}[3][\ifnumgreater{\N}{0}]{
\csvreader[
  head to column names,
  filter test=#1,
  before reading=\centering\sisetup{table-number-alignment=right},
  before first line={
    \begin{tabular} {|cc|c|c|c|c|c|c|c|c|c|c|c|}\cline{1-10}%
    \empty &\hspace{-0.6cm} \#cores & $n$ & $dim(V_0)$ & $dim(\tilde W_0)$ & setup(s) & \#It & gmres(s) & total(s) & \#It $N_S^{-1}$ \\\hline%
  },
  late after line={\\\hline},
  late after last line={\\\hline\end{tabular}}
]{#3}{}{
  \empty &
  \tablenum[table-format = 3.0, group-minimum-digits=3]{\N} &
  \tablenum[table-format = 10.0]{\dofs} &
  \tablenum[table-format = 5.0, group-minimum-digits=3]{\nCSA} &
  \tablenum[table-format = 5.0, group-minimum-digits=3]{\nCSS} &
  \tablenum[table-format = 4.1]{\setup} &
  \tablenum[table-format = 2.0]{\ItS} &
  \tablenum[table-format = 4.1]{\gmresS} &
  \tablenum[table-format = 4.1]{\totalS} &
  \tablenum[table-format = 2.0]{\MAOIt}
}
}

 \usepackage[T1]{fontenc}
\usepackage{authblk}

\newcommand{\vect}[1]{\boldsymbol{#1}}
\newcommand{\strainLinear}{\tens{\vareps}}
\newcommand{\tens}[1]{\underline{\underline{#1}}}
\newcommand{\vareps}{\varepsilon}
\newcommand{\Div}{\vect{\text{div}}\,}

\author[1]{F.~Nataf \thanks{frederic.nataf@sorbonne-universite.fr}}
\author[2]{P.H.~Tournier \thanks{tournier@ljll.math.upmc.fr}}

\affil[1,2]{\footnotesize Sorbonne Université, CNRS, Université de Paris, Inria Equipe Alpines, Laboratoire Jacques-Louis Lions, F-75005 Paris, France\\ }

\title{ A GenEO Domain Decomposition method\\ for Saddle Point problems}

\begin{document}

\maketitle

\setcounter{tocdepth}{5}
\tableofcontents
%\cleardoublepage
\pagenumbering{arabic}

\begin{abstract}
	We introduce an adaptive element-based domain decomposition (DD) method for solving saddle point problems defined as a block two by two matrix. The algorithm does not require any knowledge of the constrained space. We assume that all sub matrices are sparse and that the diagonal blocks are spectrally equivalent to a sum of positive semi definite matrices. The latter assumption enables the design of adaptive coarse space for DD methods that extends the GenEO theory \cite{Spillane:2014:ASC} to saddle point problems. Numerical results on three dimensional elasticity problems for steel-rubber structures discretized by a finite element with continuous pressure are shown for up to one billion degrees of freedom.
\end{abstract}

\section{Introduction} % (fold)
\label{sec:introduction}

Solving saddle point problems with parallel algorithms is very important for many branches of scientific computing: fluid and solid mechanics, computational electromagnetism, inverse problems and optimization.

We are interested in domain decomposition (DD) methods since they are naturally well-fitted to modern parallel architectures. For specific systems of partial differential equations with a saddle point formulation, efficient DD methods have been designed, see e.g.,~\cite{Pasciak:2002:OSM,Klawonn:1998:OPC,Pavarino:2002:BNN} and \cite{Toselli:2005:DDM} references therein.  Also in \cite{haferssas:2017:additive}, a GenEO coarse space is introduced for the P.L.~Lions' algorithm and its efficiency is mathematically proved for symmetric definite positive problems. In the above article, numerical experiments are conducted on three dimensional elasticity problems for steel-rubber structures discretized by a finite element with continuous pressure. Although the method works well in practice, the method lacks theoretical convergence guarantees and also  demands the design of specific absorbing conditions as interface conditions. 

As for a convergence rate analysis for a discretization with a continuous pressure, the recent article \cite{Widlund:2021:block} generalizes the theory developed in~\cite{Tu:2015:FDT} to the case of nonzero pressure block but under the assumption that the discontinuities are resolved by the subdomains.

Compared to the above mentioned works, the method we propose has a provable control on the condition number for zero or non zero pressure block with a continuously discretized pressure also in the case arbitrary heterogeneities and bypasses the need for absorbing boundary conditions. \\

Here as in \cite{Murphy:2000:NPI,Benzi:2005:NSS,sturler:2005:BDC,Rees:2020:EBP}, we consider the problem in the form of a two by two block matrix. Let $m$ and $n$ be two integers with $m < n$. Let $A$ $n\times n$ SPD matrix and $B$ be a sparse $m\times n$ full rank matrix of constraints and $C$ a $m\times m$ non negative matrix (in particular, $C=0$ is allowed), we consider the following saddle point matrix:
\begin{equation}
	\label{eq:saddlepointmatrix}
{\mathcal A}:=
\left(\begin{array}{cr}
	A & B^T\\
	B & -C\ \ 
\end{array}\right)\,.
\end{equation}
When the kernel of matrix $B$ is known, very efficient multigrid methods have been designed in the context of finite element methods, see e.g.,~\cite{Cahouet:1988:F3F,Hiptmair:1997:MMH,Hiptmair:1998:MMM,Arnold:2000:MHH,Reitzinger:2002:AMM,Farrell:2019:ALP}. Without this knowledge, it is nevertheless also possible to design efficient geometric multigrid methods as in~\cite{Drzisga:2018:analysis} where the fine mesh is obtained by several uniform mesh refinements.\\ 

Here we do not assume any knowledge on the kernel of matrix $B$ and we work with arbitrary meshes. The following three factor factorization, see e.g.,~\cite{Benzi:2008:SPT}:
\[
\left(\begin{array}{cc}
	A & B^T\\
	B & -C
\end{array}\right)
 = \left(\begin{array}{cc}
	\id & 0\\
	B A^{-1} & \id
\end{array}\right)
 \left(\begin{array}{cc}
	A & 0\\
	0 & - (C +\, B A^{-1}B^T)
\end{array}\right)
\left(\begin{array}{cc}
	\id & A^{-1} B^T \\
	0 & \id
\end{array}\right)\,,
\]
shows that solving the linear system with $\mathcal A$ can be performed by solving sequentially  linear systems with $A$ and one with the Schur complement $C+B A^{-1}B^T$. In order to build a scalable method, we assume that all three matrices $A$, $B$ and $C$ are sparse and that $A$ and $C$ are the sum of positive semi definite matrices. This is easily achieved in finite element or finite volume contexts for partial differential equations. The latter assumption enables the design of adaptive coarse space for DD methods, see~\cite{Dolean:2015:IDDSiam}.

More precisely, in \S~\ref{sec:preconditioning_of_matrix_a}, we recall the two-level additive Schwarz method denoted $M_A$ used to precondition the matrix $A$ (the primal-primal block of the saddle point problem). Then in \S~\ref{sec:schur_complement_preconditioning}, we introduce the operator $P_S:=C+B M_A^{-1}B^T$ which is spectrally equivalent to the Schur complement $S$. Its preconditioning is studied in \S~\ref{sec:preconditioning_of_s_2}. In \S~\ref{sec:recap}, we combine these different components to define in a compact way the parallel saddle point preconditioner. In \S~\ref{sec:num}, we present weak and strong scaling experiments on large scale elasticity problems for steel-rubber structures discretized by a finite element with continuous pressure. These problems are highly heterogeneous since the Lamé-Poisson coefficients of the rubber are $(E_1,\nu_1)=(\num{1e+7},0.4999)$ and those of the steel are $(E_2,\nu_2)=(\num{2e+9},0.35)$.     

% section introduction (end)
\section{Preconditioning of matrix $A$} % (fold)
\label{sec:preconditioning_of_matrix_a}

The sparse $n\times n$ SPD matrix $A$ is preconditioned by a two-level Schwarz type DD method:
\begin{equation}
	\label{eq:M2}
	M_A^{-1} := R_0^T\,(R_0 A R_0^T)^{-1}\, R_0 + \sum_{i=1}^N R_i^T\,(R_i A R_i^T)^{-1}\, R_i \,,
\end{equation}
where 
%$P_0:= R_0^T\,(R_0 A R_0^T)^{-1}\, R_0 A$ with 
$R_0$ is full rank $dim(V_0)\times n$ where $V_0$ denotes the space spanned by the columns of $R_0^T$. The following assumptions are crucial to ensure the final method is scalable:
\begin{assumption}[dimension and structure of the coarse space]
	\label{as:dimstructV0}
	\ 
	\begin{itemize}
		\item The coarse space dimension, $dim(V_0)$, is $O(N)$ typically 10-20 times $N$.
		\item The coarse space is made of extensions by zero of local vectors. 
	\end{itemize}
\end{assumption}
Using the GenEO method~\cite{Spillane:2014:ASC}, it is possible to fix in advance two constants $0<\lambda_m < 1 < \lambda_M$ and then build a coarse space $V_0$ such that $M_A^{-1} $ is spectrally equivalent to $A^{-1}$:
\begin{equation}
	\label{eq:spectralequivA}
   \frac{1}{\lambda_M}\, M_A^{-1}  \le A^{-1} \le \frac{1}{\lambda_m}\,M_A^{-1} \,,
\end{equation}

The dimension of the coarse space $V_0$ is typically proportional to the number of subdomains. This corresponds to Assumption~\ref{as:dimstructV0}. More precisely, for each subdomain $1\le i\le N$, let $D_i$ be a non negative diagonal matrix that defines a discrete partition of unity, i.e.:
\[
	\sum_{i=1}^N  R_i^T \, D_i\,  R_i = \id\,,
\]
and $A_i^{Neu}$ be a symmetric semi-definite positive matrix such that for the maximum multiplicity of the intersection of subdomains denoted $k_1$, we have:
\begin{equation}
	\label{eq:AiNeu}
   \sum_{i=1}^N   R_i^T A_i^{Neu} R_i  \le k_1  A\,.
\end{equation}
Then, the GenEO eigenvalue problem is local to each subdomain and reads:\\
Find $(\lambda_{ik} , V_{ik})\in \R\times \R^{rank(R_i)}$ such that:
\begin{equation}
	\label{eq:geneoA}
	(D_i R_i A R_i^T D_i)\, V_{ik} = \lambda_{ik} A_i^{Neu} V_{ik}\,.	
\end{equation}
Let $\tau>0$ be a  positive threshold, the coarse space is the vector space spanned by the vectors $R_i^T D_i V_{ik}$ for all $\lambda_{ik} > \tau$. Then inequality~\eqref{eq:spectralequivA} holds with $\lambda_m := (1+k_1\tau)^{-1} $ and $\lambda_M := k_0$ where $k_0$ is the maximal number of neighbours of a subdomain including itself.\\

Our aim in the next section is to precondition the Schur complement $-S$ of matrix ${\mathcal A}$ (eq.~\eqref{eq:saddlepointmatrix}) where 
\begin{equation}
	\label{eq:S}
\boxed{	S := C + B\,A^{-1}\,B^T\,. }
\end{equation}
This is achieved by a series of spectrally equivalent matrices or  preconditioners, the first one being $P_S$ defined as follows:
\begin{equation}
	\label{eq:S2}
\boxed{	P_S := C + B\,M_A^{-1}\,B^T\,}
\end{equation}
and the final one being $N_S^{-1}$ introduced in \S~\ref{sub:preconditioner_for_m_s}, see eq.~\eqref{eq:NS}. Finally in \S~\ref{sub:dd_solver_for_the_saddle_point_system} we will introduce the preconditioner of the saddle point matrix ${\mathcal A}$.

\section{Schur complement preconditioning} % (fold)
\label{sec:schur_complement_preconditioning}

\subsection{First spectrally equivalent preconditioner} % (fold)
\label{sub:first_spectrally_equivalent_preconditioner}

Note that $P_S$ is by definition a sum of $N+2$ positive semi definite matrices
\begin{equation}
	P_S := B\,R_0^T\,(R_0 A R_0^T)^{-1}\, R_0\,B^T\, + C + \sum_{i=1}^N B R_i^T\,(R_i A R_i^T)^{-1}\, R_i B^T\,. 
\end{equation}
Since $B$ is a sparse matrix, it is interesting to introduce, for all $0\le i\le N$, $\tilde R_i$ the restriction operator on the support of $\Im (B\, R_i^T)$ so that
\begin{equation}
\label{eq:Ritilde}
	 \tilde R_i^T \tilde R_i B R_i^T = B R_i^T.
\end{equation}
Then by defining for $0\le i\le N$,
\[
	\tilde B_i := \tilde R_i\, B\, R_i^T\,,
\]
the operator $P_S$ is rewritten as 
\begin{equation}
\label{eq:MS}
	P_S := \tilde R_0^T \tilde B_0\,(R_0 A R_0^T)^{-1}\, \tilde B_0^T \tilde R_0\, + C + \sum_{i=1}^N \tilde R_i^T \tilde B_i\,(R_i A R_i^T)^{-1}\, \tilde B_i^T \tilde R_i\,.
\end{equation}
We consider a partition of unity on $\HO:=\R^m$ defined with local diagonal matrices $(\tilde D_i)_{1\le i\le N}\in\R^{dim (Im (B\, R_i^T))\times dim (Im (B\, R_i^T))}$:
\[
	\sum_{i=1}^N \tilde R_i^T \, \tilde D_i\, \tilde R_i = \id_{\HO}\,.
\]
\begin{remark}
	This partition of unity exists since 
\[
	B = \sum_{i=1}^N B R_i^T D_i R_i = \sum_{i=1}^N \tilde R_i^T \tilde R_i  (B\,R_i^T D_i R_i)
\]
is full rank.
\end{remark}

We make the following assumption
\begin{assumption}
	\label{as:Cisasum}
	There exist symmetric positive semidefinite matrices $(\tilde C_i)_{1\le i \le N}$ such that for some constant $\tilde k_1$ 
\begin{equation}
	\label{eq:Ctildei}
	C \le \sum_{i=1}^N \tilde R_i^T\,\tilde C_i\,\tilde R_i \le  \tilde k_1\,C \,.	
\end{equation}
\end{assumption}
This assumption is not so restrictive. Indeed, for a minimization problem with constraints enforced exactly without penalization nor relaxation, we have $C=0$ and the assumption is automatically satisfied. Moreover, we have:
\begin{lemma}
	\label{th:Cdiagonalisasum}
	If $C$ is a diagonal matrix, Assumption~\ref{as:Cisasum} is satisfied with $\tilde k_1=1$.
\end{lemma}
\begin{proof}
If $C$ is a diagonal matrix, it suffices to take
\[
	\tilde C_i := \tilde R_i C \tilde R_i^T \tilde D_i\,,
\]
which is a diagonal non negative matrix. Indeed, we then have for all ${\mathbf P}\in \R^m$:
\[
	C\, {\mathbf P} = \sum_{i=1}^N C\, \tilde R_i^T \tilde D_i \tilde R_i \,{\mathbf P}
					= \sum_{i=1}^N \tilde R_i^T (\tilde R_i C\, \tilde R_i^T \tilde D_i) \tilde R_i \,{\mathbf P}\,.
\]
\end{proof}
\begin{remark}
Note also that in the finite element case it suffices to restrict to the subdomains the variational form that defines $C$. In this case $\tilde k_1$ is the multiplicity of the intersections of the subdomains used to define the $\tilde C_i$'s.
\end{remark}
Let us define the operator $M_S$ as the sum of a non local but low rank matrix $S_0$:
\[
	S_0 := \tilde R_0^T \tilde B_0\,(R_0 A R_0^T)^{-1}\, \tilde B_0^T \tilde R_0\,,
\]
 and of $S_1$ which is a sum of $N$ local positive semi definite matrices:
\[
\boxed{	S_1 := \sum_{i=1}^N \tilde R_i^T (\tilde C_i + \tilde B_i\,(R_i A R_i^T)^{-1}\, \tilde B_i^T) \tilde R_i\,,}
\]
that is 
\[
\boxed{	M_S := S_0 + S_1\,.}
\]
By Assumption~\ref{as:Cisasum}, the operator $M_S$ is spectrally equivalent to $P_S$ which is also spectrally equivalent to $S$. Note that we may assume that $S_1$ is invertible whereas it does not make sense for $S_0$. Note that if it is not the case, since we build a preconditioner, $S_1$ can be regularized by a small diagonal term with little effect on the efficiency of the preconditioner.\\

We consider next the construction of a preconditioner $M_{S_1}^{-1}$ to $S_1$ leveraging the fact that $S_1$ is a sum of symmetric semidefinite positive matrices. Let us stress that this property stems from the domain decomposition structure of the preconditioner for matrix $A$ which apart from the coarse level is block diagonal.

% subsection first_spectrally_equivalent_preconditioner (end)

% section schur_complement_preconditioning (end)

\subsection{Preconditioning of $S_1$} % (fold)
\label{sec:preconditioning_of_s_2}

It is well known that one level domain decomposition methods are in most cases not scalable. Nevertheless, the study of a one-level method in \S~\ref{sec:One-level DD} enables the identification of a suitable coarse space that will be efficiently embedded in a scalable two-level domain decomposition method in \S~\ref{sec:twolevelMS1}.

Our studies of the spectrum of the DD preconditioners are based on the Fictitious Space lemma which is recalled here, see \cite{Nepomnyaschikh:1991:MTT} for the original paper and \cite{Griebel:1995:ATA} for a modern presentation. 
 \begin{lemma}[Fictitious Space Lemma, Nepomnyaschikh 1991]
   \label{th:fictitiousSpaceLemma}  
   Let $\HO$ and $\HP$ be two Hilbert spaces, with the scalar products
   denoted by $(\cdot,\cdot)$ and $(\cdot,\cdot)_D$. Let the symmetric positive bilinear forms $a\,:\,\HO \times \HO \rightarrow \R$ and $b\,:\,\HP \times \HP \rightarrow \R$, generated by the s.p.d. operators $A\,:\,\HO \rightarrow \HO$ and $B\,:\,\HP \rightarrow \HP$, respectively (i.e. $(Au,v)=a(u,v)$ for all $u,v\in\HO$ and $(Bu_D,v_D)_D=b(u_D,v_D)$ for all $u_D,v_D\in\HP$). Suppose that there exists a linear operator $\RL\,:\,\HP \rightarrow \HO$ that satisfies the following three assumptions:
\begin{itemize}
	\item[(i)] $\RL$ is surjective.
	\item[(ii)] Continuity of $\RL$: there exists a positive constant $c_R$ such that 
	\begin{equation}\label{eq:cr}
	a(\RL u_D,\RL u_D) \le c_R\cdot  b(u_D,u_D)\ \ \forall u_D\in \HP\,.
	\end{equation}
	\item[(iii)] Stable decomposition: there exists a positive constant $c_T$ such that for all $u\in\HO$ there exists $u_D\in \HP$ with $\RL u_D=u$ and 
	\begin{equation}\label{eq:ct}
	c_T\cdot b(u_D,u_D) \le a(\RL u_D,\RL u_D)=a(u,u)\,.
	\end{equation}	
\end{itemize}
We introduce the adjoint operator $\RL^*\,:\,\HO\rightarrow \HP$ by
$(\RL u_D,\,u) = (u_D,\,\RL^* u)_D$ for all $u_D\in\HP$ and
$u\in\HO$.\\
 Then, we have the following spectral estimate
\begin{equation}\label{eq:fictestimate}
%\boxed{
c_T\cdot  a(u,u) \le a\left(\RL B^{-1} \RL^* A u,\,u\right) \le c_R\cdot  a(u,u)\,,\ \ \forall u\in\HO\,
%} 
\end{equation}
which proves that the eigenvalues of operator $\RL B^{-1} \RL^* A$ are bounded from below by $c_T$ and from above by $c_R$. 
 \end{lemma}
The Fictitious Space Lemma (FSL) can also conveniently be related to the book~\cite{Toselli:2005:DDM}: the first assumption corresponds to equation~(2.3), page 36 where the global Hilbert space is assumed to satisfy a decomposition into subspaces, the second assumption is related to Assumptions~2.3 and 2.4,  page~40 and the third assumption corresponds to the Stable decomposition Assumption~2.2 page~40.\\

%\subsubsection{Equivalent preconditioners of $S_1$} % (fold)
\label{sub:equivalent_preconditioner_on_h}
\subsubsection{One-level DD for $S_1$}
	\label{sec:One-level DD}
As in~\cite{Dolean:2015:IDDSiam} chapter~7, we begin with a one-level Neumann-Neumann type DD method defined in terms of the Fictitious Space Lemma (FSL). This study will be the basis for constructing the two-level preconditioner in \S~\ref{sec:twolevelMS1}. Recall the formula for the one-level preconditioner $M_{S_1,1}^{-1}$ for $S_1$:
\begin{equation}
	\label{eq:MS1onelevel}
	M_{S_1,1}^{-1} := \sum_{i=1}^N \tilde R_i^T \tilde D_i\, (\tilde C_i + \tilde B_i\,(R_i A R_i^T)^{-1}\,\tilde B_i^T)^\dag\, \tilde D_i \tilde R_i\,,
\end{equation}
where the superscript $\dag$ denotes a pseudo inverse in case the operator in brackets is not invertible. For sake of simplicity, we assume that they are invertible so that the following framework enables the study of $M_{S_1,1}$ with the fictitious space lemma. Let 
\[
	\HO:=\R^m
\]
 and let $a$ be the following bilinear form:
 \[
 	a:\HO \times \HO \rightarrow \R  \ \ a({\mathbf P},{\mathbf Q}) := (S_1{\mathbf P},{\mathbf Q} )\,.
\]
Let 
\[
	\HP := \Pi_{i=1}^N \R^{rank (\tilde B_i)}\,,
\]
and $b$ be the following bilinear form:
 \begin{equation}
	 \label{eq:bS1-onelevel}
 	b:\HP \times \HP \rightarrow \R
  \ \ b( ({\mathbf P}_i)_{1\le i\le N} , ({\mathbf Q}_i)_{1\le i\le N} ) := 
  \sum_{i=1}^N (\,(\tilde C_i + \,\tilde B_i\,(R_i A R_i^T)^{-1}\,\tilde B_i^T)\,{\mathbf P}_i\, ,\,{\mathbf Q}_i )\,.
 \end{equation}
We define $\RL$:
\[
	\begin{array}{rcl}
		\RL: &\HP \rightarrow & \HO \\
		     & ({\mathbf P}_i)_{1\le i \le N} \mapsto & \sum_{i=1}^N \tilde R_i^T \tilde D_i {\mathbf P}_i\,.
	\end{array}
\]

We now check the three assumptions of the FSL.
\paragraph{Surjectivity of $\RL$}
For any ${\mathbf P} \in\HO$, we have: 
\[
	{\mathbf P} =  \sum_{i=1}^N \tilde R_i^T \tilde D_i \tilde R_i {\mathbf P}\,,
\]
so that 
\begin{equation}
	\label{eq:decompositiononelevel}
{\mathbf P} = \RL( (\tilde R_i {\mathbf P} )_{1\le i\le N})\,.
\end{equation}

\paragraph{Continuity of $\RL$} % (fold)
\label{par:continuity_of_rlonelevel}
On one hand, we have using $\NC$ the number of neighbours of a subdomain (including itself), $\NC:= \max_{1\le i\le N} \# {\mathcal O}(i) $ where  ${\mathcal O}(i) := \{1 \le j \le N \ | \  \tilde R_i\,\tilde D_i\, S_1\,\tilde D_j\,  \tilde R_j^T \neq 0\}$:
\[
	\begin{array}{rcl}
		a(\RL({\mathcal P}) \,,\, \RL({\mathcal P})) 
		&=&  \| (\sum_{i=1}^N \tilde R_i^T \tilde D_i {\mathbf P}_i)\|_a^2 
		\le \NC\, \sum_{i=1}^N  \| \tilde R_i^T \tilde D_i {\mathbf P}_i \|_a^2 \\
		&=& \NC\, ( \left(\sum_{j\in {\mathcal O}(i)} \tilde R_j^T (\tilde C_j + \tilde B_j\,(R_j A R_j^T)^{-1}\, \tilde B_j^T) \tilde R_j \right)  \tilde R_i^T \tilde D_i {\mathbf P}_i \,,\, \tilde R_i^T \tilde D_i {\mathbf P}_i) \,.
	\end{array}
\]
On the other hand, we have by definition:
\[
	b({\mathcal P}\,,\,{\mathcal P}) =  \sum_{i=1}^N (\, (\tilde C_i + \tilde B_i\,(R_i A R_i^T)^{-1}\,\tilde B_i^T)\,{\mathbf P}_i\, ,\,{\mathbf P}_i )\,.
\]
We can take:
\begin{equation}
	\label{eq:cRMS1onelevel}
	c_R := \max_{1\le i\le N} \max_{P_i\in\R^{rank(\tilde B_i)}} \frac{( \sum_{j\in \mathcal{O}(i)} \tilde R_i\,\tilde R_j^T\, (\tilde C_j + \tilde B_j\,(R_j A R_j^T)^{-1}\, \tilde B_j^T)\, \tilde R_j \,  \tilde R_i^T \tilde D_i {\mathbf P}_i \,,\,  \tilde D_i {\mathbf P}_i) }{(\,(\tilde C_i + \tilde B_i\,(R_i A R_i^T)^{-1}\,\tilde B_i^T)\,{\mathbf P}_i\, ,\,{\mathbf P}_i )} \,,	
\end{equation}
but we have no control on $c_R$ which may be large. This motivates the introduction of a spectral coarse space in \S~\ref{sec:twolevelMS1} with the generalized eigenvalue problem~\eqref{eq:geneoforS1} .
% paragraph continuity_of_rl (end)
\paragraph{Stable decomposition} % (fold)
\label{par:stable_decompositiononelevel}
Let ${\mathbf P}\in\HO$, we start from its decomposition~\eqref{eq:decompositiononelevel} and estimate its $b$-norm
\[
	\begin{array}{rcl}
 b({\mathcal P}\,,\,{\mathcal P}) &=& \sum_{i=1}^N (\, (\tilde C_i + \tilde B_i\,(R_i A R_i^T)^{-1}\,\tilde B_i^T)\, \tilde R_i {\mathbf P} \, ,\,\tilde R_i {\mathbf P} ) = a( {\mathbf P} \, ,\, {\mathbf P} )\,,
	\end{array}
\]		
so that we can take $c_T=1$ which is an optimal value.
% paragraph stable_decomposition (end)

\subsubsection{Two-level DD for $S_1$}
	\label{sec:twolevelMS1}
In order to control the value of $c_R$ defined above, we introduce two two-level preconditioners. The first one is similar to what is done for Schur complement methods in \cite[\S~7.8.3, page 197]{Dolean:2015:IDDSiam} or in \cite{Spillane:2013:ASC}. The second one is a cheaper lightweight version of the former but then not with a full control of its spectrum. In practice, in our numerical experiments both methods perform similarly in terms of iterations counts and thus with an advantage in terms of elapsed time for the second preconditioner. For both two-level methods, the generalized eigenvalue value problem in each subdomain $1\le i\le N$ to be solved to build the coarse space is inferred from the definition of the constant $c_R$ in eq.~\eqref{eq:cRMS1onelevel}:
\begin{equation}
	\label{eq:geneoforS1}
	\begin{array}{r}
\tilde D_i\,\tilde R_i\,   \left( \sum_{j\in {\mathcal O}(i)}  \tilde R_j^T (\tilde C_j + \tilde B_j\,(R_j A R_j^T)^{-1}\, \tilde B_j^T) \tilde R_j \right) \,\,\tilde R_i^T  \tilde D_i \ {\mathbf P}_{i\,k}  \\
=  \lambda_{i\,k} (\tilde C_i + \tilde B_i\,(R_i A R_i^T)^{-1}\,\tilde B_i^T)\ {\mathbf P}_{i\,k} \,.
	\end{array}
\end{equation}

This generalized eigenvalue problem contains inverses of some local matrices on both sides and in order to solve it via e.g. Arpack, we have to factorize one of them. This difficulty can be overcome since the right matrix is the inverse of the Schur complement of the extended sparse matrix~\eqref{eq:systemeetendu}. It is thus amenable to a factorization using only sparse matrix factorizations.

It can be solved in $O(1)$ communications. The coarse space is defined as follows. Let $\tau_{S_1}$ be a user-defined threshold; for each subdomain $1 \le i \le N$, we introduce a subspace $\tilde W_i\subset \R^{rank(\tilde B_i)}$:
\begin{equation}
	\label{eq:BSS1i}
	\tilde W_i := \text{Span}\{ {\mathbf P}_{i\,k}\ |\ \lambda_{i\,k} > \tau_{S_1}  \}\,.
\end{equation} 
Then the coarse space $\tilde W_0$ is defined (with some abuse of notation)
\[
\boxed{	\tilde W_0 := \bigoplus_{1\le i\le N} \tilde R_i^T \tilde D_i \tilde W_i \,.}
\]
Let $Z_{S_1}$ be a rectangular matrix whose columns span the coarse space $\tilde W_0$. Let $\tilde P_0$ be the $S_1$ orthogonal projection from $\R^m$ on $\tilde W_0$ whose formula is 
\begin{equation}
	\label{eq:P0tilde}
	\tilde P_0 = Z_{S_1} (Z_{S_1}^T S_1 Z_{S_1})^{-1} Z_{S_1}^T S_1\,.
\end{equation}
In order to avoid a too cumbersome analysis, we make the following assumption:
\begin{assumption}
	\label{eq:Siinvertible}
We assume that for all subdomains $1\le i \le N$, $\tilde C_i + \tilde B_i\,(R_i A R_i^T)^{-1}\,\tilde B_i^T$ is invertible.
\end{assumption}
Finally, the first preconditioner for $S_1$ reads
\begin{equation}
	\label{eq:S1prec}
	\begin{array}{rcl}
  M_{S_1}^{-1} &:=& Z_{S_1}\,(Z_{S_1}^T {S_1} Z_{S_1})^{-1}\,Z_{S_1}^T + (\id - \tilde P_0)\\
   & \times & \left(\sum_{i=1}^N \tilde R_i^T \tilde D_i\, (\tilde C_i + \tilde B_i\,(R_i A R_i^T)^{-1}\,\tilde B_i^T)^\dag\, \tilde D_i \tilde R_i \right) (\id-  \tilde P_0^T)\,.		
	\end{array}
\end{equation}
If Assumption~\ref{eq:S1prec} is not satisfied for some subdomain $i$, we should incorporate the kernel of $\tilde C_i + \tilde B_i\,(R_i A R_i^T)^{-1}\,\tilde B_i^T$ in the coarse space and make use of a pseudoinverse in the definition of the preconditioner as it is done for the FETI method, see~\cite{Farhat:1991:MFE} or \cite{Dolean:2015:IDDSiam} \S~7.8.2 and references therein.

Recall that from~\cite{Dolean:2015:IDDSiam} chapter~7, we have for $\alpha := \max(1 , \frac{\NC}{\tau_{S_1}})$:
\[
  \frac{1}{\alpha} M_{S_1}^{-1} 	\le S^{-1} \le M_{S_1}^{-1} \,.
\]

A careful implementation of \eqref{eq:S1prec} requires two coarse solves. In order to simplify the application of $M_{S_1}^{-1}$, we can save one coarse solve by proposing an alternative definition of preconditioner $M_{S_1}^{-1}$ without the global projection $\tilde P_0$. We directly define it using the FSL framework. We keep definitions for $\HO$ and $a$ from the beginning of \S~\ref{sec:One-level DD}. But now the space $\HP$ is defined as:
\[
	\HP := \R^{rank (Z_{S_1})} \times \Pi_{i=1}^N \R^{rank (\tilde B_i)}\,,
\]
the operator $\RL$ is defined using for $1\le i \le N$ the $(\tilde C_i + \tilde B_i\,(R_i A R_i^T)^{-1}\, \tilde B_i^T)$ orthogonal projections $\mathbf{\xi_i}$ on $\tilde W_i$ and parallel to $\text{Span}\{ {\mathbf P}_{i\,k}\ |\ \lambda_{i\,k} \le \tau_{S_1}  \}$ :
\[
	\begin{array}{rcl}
		\RL: &\HP \rightarrow & \HO \\
		     &  ({\mathbf P}_i)_{0\le i \le N} \mapsto & Z_{S_1} {\mathbf P}_0 + \sum_{i=1}^N \tilde R_i^T \tilde D_i (\id-\xi_i) {\mathbf P}_i\,,
	\end{array}
\]
and $b$ is the following bilinear form:
 \[
 	b:\HP \times \HP \rightarrow \R\,,
  \ \ b( ({\mathbf P}_i)_{0\le i \le N} , ({\mathbf Q}_i)_{0\le i \le N} ) := (S_1 Z_{S_1} {\mathbf P}_0 , Z_{S_1} {\mathbf Q}_0 ) +
  \sum_{i=1}^N (\,(\tilde C_i + \,\tilde B_i\,(R_i A R_i^T)^{-1}\,\tilde B_i^T)\,{\mathbf P}_i\, ,\,{\mathbf Q}_i )\,.
\]
We now check the three assumptions of the FSL.
\paragraph{Surjectivity of $\RL$}
For any ${\mathbf P} \in\HO$, we have: 
\[
	{\mathbf P} =  \sum_{i=1}^N \tilde R_i^T \tilde D_i \xi_i \tilde R_i {\mathbf P}
	+ \sum_{i=1}^N \tilde R_i^T \tilde D_i (\id-\xi_i) \tilde R_i {\mathbf P}\,.
\]
Note that since the first sum belongs to $\tilde W_0$, there exists ${\mathbf P}_0$ such that $\sum_{i=1}^N \tilde R_i^T \tilde D_i \xi_i \tilde R_i {\mathbf P}=Z_{S_1} {\mathbf P}_0$. Then, we have:
\begin{equation}
	\label{eq:decompositiontwolevel2}
{\mathbf P} = \RL( {\mathbf P}_0 , (\tilde R_i {\mathbf P} )_{1\le i\le N})\,.
\end{equation}
\paragraph{Continuity of $\RL$} 
For ${\mathcal P}\in\HP$ we have:
\[
	\begin{array}{rcl}
		a(\RL({\mathcal P}) \,,\, \RL({\mathcal P})) 
		&=& (S_1 \RL({\mathcal P}) \,,\, \RL({\mathcal P})) \\
		&\le& 2\,( (S_1 Z_{S_1} {\mathbf P}_0 \,,\, Z_{S_1} {\mathbf P}_0))
		 + (S_1 \sum_{i=1}^N (\id-\xi_i) \tilde R_i^T \tilde D_i {\mathbf P}_i \,,\, \sum_{i=1}^N (\id-\xi_i) \tilde R_i^T \tilde D_i {\mathbf P}_i)  ) \\
		&\le& 2\,( (S_1 Z_{S_1} {\mathbf P}_0 \,,\, Z_{S_1} {\mathbf P}_0))
		 + \NC \sum_{i=1}^N (S_1 (\id-\xi_i) \tilde R_i^T \tilde D_i {\mathbf P}_i \,,\,  (\id-\xi_i) \tilde R_i^T \tilde D_i {\mathbf P}_i)  ) \\
		 &\le& 2\,( (S_1 Z_{S_1} {\mathbf P}_0 \,,\, Z_{S_1} {\mathbf P}_0))
		 + \NC \tau_{S_1} \sum_{i=1}^N ((\tilde C_i + \tilde B_i\,(R_i A R_i^T)^{-1}\,\tilde B_i^T)  {\mathbf P}_i \,,\, {\mathbf P}_i)  ) \\
		 &\le& 2 \max(1 , \NC \tau_{S_1})\, b({\mathcal P}\,,\,{\mathcal P})  \,,
	\end{array}
\]
so that we can take $c_R=2 \max(1 , \NC \tau_{S_1} )$ and we lose only a factor of 2 compared to eq.~\eqref{eq:S1prec}.

\paragraph{Stable decomposition} % (fold)
\label{par:stable_decompositiontwolevel}
Let ${\mathbf P}\in\HO$, we start from its decomposition~\eqref{eq:decompositiontwolevel2} and estimate its $b$-norm
\[
	\begin{array}{rcl}
 b({\mathcal P}\,,\,{\mathcal P}) &=& (S_1 \sum_{i=1}^N \tilde R_i^T \tilde D_i \xi_i \tilde R_i {\mathbf P} , \sum_{i=1}^N \tilde R_i^T \tilde D_i \xi_i \tilde R_i {\mathbf P} ) \\
 & &+ \sum_{i=1}^N (\, (\tilde C_i 
 + \tilde B_i\,(R_i A R_i^T)^{-1}\,\tilde B_i^T)\,  \tilde R_i {\mathbf P} \, ,\, \tilde R_i {\mathbf P} ) \\
 &=& (S_1 \sum_{i=1}^N \tilde R_i^T \tilde D_i \xi_i \tilde R_i {\mathbf P} , \sum_{i=1}^N \tilde R_i^T \tilde D_i \xi_i \tilde R_i {\mathbf P} ) + a({\mathbf P}\,,\,{\mathbf P})  \\
   &\le& (\gamma + 1)\,a({\mathbf P}\,,\,{\mathbf P})
   \,,
	\end{array}
\]		
where 
\[
	\gamma := \max_{{\mathbf P} } \frac{(S_1 \sum_{i=1}^N \tilde R_i^T \tilde D_i \xi_i \tilde R_i {\mathbf P} , \sum_{i=1}^N \tilde R_i^T \tilde D_i \xi_i \tilde R_i {\mathbf P} ) }{(S_1{\mathbf P}\,,\,{\mathbf P})}\,.
\]
We can take $c_T:= 1/(1+\gamma)$ {but we have no estimate on $\gamma$}. 

Finally, the explicit form of this alternative preconditioner reads:
\begin{equation}
	\label{eq:S1preclight}
	\begin{array}{rcl}
  M_{S_1}^{-1} &:=& Z_{S_1}\,(Z_{S_1}^T {S_1} Z_{S_1})^{-1}\,Z_{S_1}^T \\
   & & + \left(\sum_{i=1}^N \tilde R_i^T \tilde D_i\,(\id-\xi_i) (\tilde C_i + \tilde B_i\,(R_i A R_i^T)^{-1}\,\tilde B_i^T)^\dag\, (\id-\xi_i^T) \tilde D_i \tilde R_i \right) \,.		
	\end{array}
\end{equation}

% paragraph stable_decomposition (end)
% section essai2 (end)

Note that the application of $M_{S_1}^{-1}$ can be done using only sparse solvers since	solving a linear system with a local Schur complement 
\[
	(\tilde C_i + \tilde B_i\,(R_i A R_i^T)^{-1}\,\tilde B_i^T) {\mathbf P}_i = {\mathbf G}_i \,,
\]	
amounts to solving an augmented sparse system which has the form of a local saddle point system:
\begin{equation}
	\label{eq:systemeetendu}
	 - \left(\begin{array}{lr}
		 	R_i A R_i^T & \tilde B_i^T \\ \tilde B_i & - \tilde C_i
		 \end{array}\right)\,
	\left(\begin{array}{c}
		 	{\mathbf U}_i \\ {\mathbf P}_i
		 \end{array}\right)
 =  \left(\begin{array}{c}
	 	{\mathbf 0} \\ {\mathbf G}_i
	 \end{array}\right)\,.
\end{equation}

% subsection equivalent_preconditioner_on_h (end) 

\subsection{Final Preconditioner for the Schur complement} % (fold)
\label{sub:preconditioner_for_m_s}
From the spectrally equivalent preconditioner $M_{S_1}$ to $S_1$, we define $N_S$ a spectrally equivalent preconditioner to $M_S$ and thus to $S$ as well:
\begin{equation}
	\label{eq:NS}
\boxed{	N_S := S_0 + M_{S_1}\,.}
\end{equation}
We now consider the application of the preconditioner $N_S$, that is for some right hand side ${\mathbf G}\in\R^m$, the solving in ${\mathbf P}$ of the following system:
\begin{equation}
	\label{eq:newinner}
		N_S\, {\mathbf P}={\mathbf G}\,,
\end{equation}
by a Krylov solver with $M_{S_1}^{-1}$ as a preconditioner.

\section{Recap} % (fold)
\label{sec:recap}
\subsection{Setup for the Schur complement preconditioner} % (fold)
\label{sub:setup}

% subsection setup (end)
We have a setup phase which is composed of:
\begin{enumerate}
	\item Build the two-level preconditioner $M_A^{-1}$ for $A$, see eq.~\eqref{eq:M2},
	\item Build the two-level preconditioner $M_{S_1}^{-1}$ for $S_1$, see eq.~\eqref{eq:S1prec}.
%	\item OPTIONAL en fait:  Compute the entries of matrix $M_{A_0}$, see eq.~\eqref{eq:shermannMorrisonS} and factorize it.
\end{enumerate}

Once the setup is complete, applying preconditioner $N_S^{-1}$ can be performed following  Algorithm~\ref{alg:NS1preconditioner}
\begin{algorithm}
\caption{$N_S^{-1}$ matvec product}
\label{alg:NS1preconditioner}
\begin{algorithmic}
\STATE 	INPUT: ${\mathbf G}\in \R^m$ \ \ OUTPUT: ${\mathbf P} = N_S^{-1}\,{\mathbf G}$
%\STATE 1. Compute $ =  {\mathbf G}' := M^{-1}_{S_1} {\mathbf G}$
\STATE 1. Solve eq.~\eqref{eq:newinner} in ${\mathbf P}$ by a Krylov method with $M_S^{-1}$ as preconditioner.
\end{algorithmic} 
\end{algorithm} 
% or by Algorithm~\ref{alg:NS1preconditionerMAO}.
% \begin{algorithm}
% \caption{$N_S^{-1}$ matvec product}
% \label{alg:NS1preconditionerMAO}
% \begin{algorithmic}
% \STATE 	INPUT: ${\mathbf G}\in \R^m$ \ \ OUTPUT: ${\mathbf P} = N_S^{-1}\,{\mathbf G}$
% \STATE 1. Compute $ =  {\mathbf G}' := M^{-1}_{S_1} {\mathbf G}$
% \STATE 2. Compute the right hand side $L_0^{-1} \, \tilde B_0^T \tilde R_0 {\mathbf G}'$ of eq.~\eqref{eq:shermannMorrisonS}
% \STATE 3. Solve eq.~\eqref{eq:shermannMorrisonS} in ${\mathbf y}$
% \STATE 4. Compute ${\mathbf P} := M^{-1}_{S_1} ({\mathbf G} - \tilde R_0^T \tilde B_0\, {L_0^T}^{-1}  {\mathbf y})$
% \end{algorithmic}
% \end{algorithm}

\subsection{DD solver for the saddle point system} % (fold)
\label{sub:dd_solver_for_the_saddle_point_system}
\begin{algorithm}
\caption{DD saddle point solver}
\label{alg:ddsaddlepointsolver}
\begin{algorithmic}
\STATE 	INPUT: $\left(\begin{array}{c}
	{\mathbf F_U} \\
	{\mathbf F_P}
	\end{array}\right) \in \R^{n+m}$
	\ \ \ \hfill 
	OUTPUT: $\left(\begin{array}{c}
		{\mathbf U} \\
		{\mathbf P}
	\end{array}\right)$ 
	 the solution to~\eqref{eq:saddlepointsystem}.
\STATE 1. Solve $A {\mathbf G_U} = {\mathbf F_U} $ by a PCG with $M_A^{-1}$ as a preconditioner
\STATE 2. Compute ${\mathbf G_P} := {\mathbf F_P} - B\, {\mathbf G_U}$
\STATE 3. Solve $(C+BA^{-1}B^T) {\mathbf P}= -{\mathbf G_P}$ by a PCG with $N_S^{-1}$ as a preconditioner, see~Algorithm~\ref{alg:NS1preconditioner}
\STATE 4. Compute ${\mathbf G_U} := {\mathbf F_U} - B^T {\mathbf P}$
\STATE 5. Solve $A {\mathbf U}  = {\mathbf G_U}$ by a PCG with $M_A^{-1}$ as a preconditioner
\end{algorithmic} 
\end{algorithm} 

We now consider the solving of the saddle point problem:
\begin{equation}
	\label{eq:saddlepointsystem}
\left(\begin{array}{cr}
	A & B^T\\
	B & -C\ \ 
\end{array}\right)\,\left( \begin{array}{c}
	{\mathbf U}\\
	{\mathbf P}
\end{array} \right)
 = \left( \begin{array}{c}
	{\mathbf F_U}\\
	{\mathbf F_P}
\end{array} \right)\,
\end{equation}
by Algorithm~\ref{alg:ddsaddlepointsolver} whose Step~3 demands the matrix-vector product with matrix $C+BA^{-1}B^T$ which is done by an iterative solve for matrix $A^{-1}$. 
In order to avoid such a nested loop, block solvers have been developed, see \cite{Murphy:2000:NPI,Benzi:2005:NSS,sturler:2005:BDC,Notay:2019:CSI}. Following~\cite{Notay:2019:CSI},   solving with $A^{-1}$ is not  needed when preconditioning the original system~\eqref{eq:saddlepointmatrix} by $P^{-1}$ where $P$ is an inexact block factorization:
\begin{equation}
	\label{eq:inexactblockfacto}
	P:= \left(\begin{array}{cr}
	M_A & \\
	B & -N_S\ \ 
\end{array}\right)\,
\left(\begin{array}{cr}
	\id & M^{-1}_A\,B^T\\
	 & \id\ \ 
\end{array}\right)\,
\end{equation}
or a slightly different but symmetric variant of it introduced in eq.~(2.5) of \cite{Notay:2019:CSI}:
\begin{equation}
	\label{eq:inexactblockfactovariant}
	P:= \left(\begin{array}{cr}
	\id & \\
	B M^{-1}_A & \id \ \ 
\end{array}\right)\,
\left(\begin{array}{cr}
	M_A\, (2\,M_A^{-1}-A)^{-1} \,M_A &  \\
	 & -N_S \ \ 
\end{array}\right)\,
\left(\begin{array}{cr}
	\id &  M^{-1}_A\,B^T \\
	 & \id \ \ 
\end{array}\right)\,.
\end{equation}
 Note that in both cases, the inverses of $P$ are computable easily in our framework.
% 
% For instance, the inverse of the last preconditioner reads:
% \begin{equation}
% 	\label{eq:inverseinexactblockfactovariant}
% 	P^{-1}:=
% \left(\begin{array}{cr}
% 	\id &  -M^{-1}_A\,B^T \\
% 	 & \id \ \
% \end{array}\right)\,
% \left(\begin{array}{cr}
% 	M_A^{-1}\, (2\,M_A^{-1}-A) \,M_A^{-1} &  \\
% 	 & -N_S^{-1} \ \
% \end{array}\right)\,
% 	\left(\begin{array}{cr}
% 	\id & \\
% 	-B M^{-1}_A & \id \ \
% \end{array}\right)\,.
% \end{equation}
% 

% \section{Variants} % (fold)
% \label{sec:variants}
% In \S~\eqref{sec:schur_complement_preconditioning}, we start with a two level additive Schwarz method (ASM)  eq.~\eqref{eq:M2} as a preconditioner for matrix $A$ in~\eqref{eq:saddlepointmatrix}. Another possibility is to start from a balancing Neumann-Neumann (BNN) or more generally a SORAS~\cite{Haferssas:ASM:2017} or BDD-H~\cite{Kimn:2007:ROB} type method:
% \begin{equation}
% 	\label{eq:MA_SORAS}
% 	M_{A_{SORAS}}^{-1} := R_0^T\,(R_0 A R_0^T)^{-1}\, R_0 + \sum_{i=1}^N  R_i^T\,D_i\,{A_i^{Rob}}^{-1}\, D_i\,R_i \,,
% \end{equation}
% where for each subdomain $1\le i\le N$, $A_i^{Rob}$ is a local Neumann matrix (BNN algorithm) or an arbitrary invertible matrix. In order to build the preconditioners for matrix $S$ and then ${\mathcal A}$, it is sufficient to replace matrices $(R_i\,A\,R_i^T)^{-1}$ with ${A_i^{Rob}}^{-1}$ in the above sections.
%
% \newpage

\section{Numerical experiments} % (fold)
\label{sec:num}

%\begin{tabular}{|l|c|c|}\hline%
%$N$ & $n$ & $n$
%\csvreader[respect sharp, head to column names]{nnew_weak_ma0Iter1e-2elasticity3d_tauA_01_tauS_03.csv}{}%
%{\\\N & \dofs & \ItS}%
%\\\hline
%\end{tabular}

In this section, we perform 3D experiments to illustrate the theory and the performance of the method. We are interested in a heterogeneous elasticity problem with nearly incompressible material typically rubber-steel structures. First, we recall the definition of the coefficients and the corresponding variational formulation.

The mechanical properties of a solid are characterized by its Young modulus $E$ and Poisson ratio $\nu$, or alternatively by its Lamé coefficients $\lambda$ and $\mu$. They verify the following relations:

\begin{equation}
	\label{eq:lambdanu}
\lambda = \frac{E \nu}{(1+\nu) (1- 2 \nu)}\ \text{ and } \ \mu = \frac{E}{2(1+\nu)}\,.
\end{equation}

For the discretization, we choose a continuous pressure space and take the lowest order Taylor-Hood finite element $C0P2-C0P1$ whose stability is proved, see e.g., \cite{Brezzi:2012:MHF}. The domain $\Omega$ is a beam and the variational problem consists in finding $(\boldsymbol{u}_{h},p_{h})\in \mathcal{V}_{h}:=\mathbb{P}_{2}^{3}\cap (H^1_0(\Omega))^3\times\mathbb{P}_{1}$ with Dirichlet boundary conditions on the four lateral faces and Neumann boundary conditions on the other two faces such that for all $(\boldsymbol{v}_{h},q_{h})\in \mathcal{V}_{h}$
\begin{equation}\label{eq:elastvarf}
\begin{cases}
	\int_\Omega 2\mu\strainLinear (\boldsymbol{u}_{h}):\strainLinear(\boldsymbol{v}_{h})dx & 
	- \int_\Omega p_{h}\Div(\boldsymbol{v}_{h})dx=
	\int_\Omega\vect{f}\boldsymbol{v}_{h} dx\\\\
	-\int_\Omega\Div(\boldsymbol{u}_{h})q_{h}dx  &- \int_\Omega \frac{1}{\lambda}p_{h}q_{h}=0.
	\end{cases}
\end{equation}
Letting $\boldsymbol{u}$ denote the degrees of freedom of $\boldsymbol{u}_{h}$ and $p$ that of $p_h$, the problem can be written in matrix form as:

\begin{equation}
	\label{eq:elastsystem}
\left(\begin{array}{cr}
	A & B^T\\
	B & -C\ \ 
\end{array}\right)\,\left( \begin{array}{c}
	{\boldsymbol{u}}\\
	{p}
\end{array} \right)
 = \left( \begin{array}{c}
	{\text{f}}\\
	{0}
\end{array} \right)\,.
\end{equation}

The matrix $C$ is a mass matrix arising from the discretization of a variational form. This enables us to satisfy Assumption~\ref{as:Cisasum} with $\tilde C_i$ the corresponding mass matrix but only defined on subdomain $\tilde \Omega_i$ which is the extension by a layer of direct neighbors of subdomain $\Omega_i$.

In the following numerical experiments, we consider a heterogeneous beam composed of 10 alternating layers of rubber material $(E_1,\nu_1)=(\num{1e+7},0.4999)$ and steel material $(E_2,\nu_2)=(\num{2e+9},0.35)$, see Fig.~\ref{fig:beam}.\\

\subsection{Software, hardware, implementation details}

In the following numerical experiments, the iteration counts assess the control of the condition number via the adaptive coarse spaces. We also report timings along with iteration counts. We also showcase the fact that the size of the GenEO coarse space adapts automatically to the difficulty of the problem at hand, for example when going from a homogeneous to a heterogeneous problem. We illustrate the efficiency of the method by performing weak and strong scalability tests, using the automatic graph partitioner \textit{Metis}~\cite{METIS} for the subdomain partitioning.\\

The problem is discretized and solved with the open-source parallel finite element software \textit{FreeFEM}~\cite{Hecht:2012:NDF}. FreeFEM is domain specific language (DSL) where the problem to be solved is defined in terms of its variational formulation. Then the local matrices $(A_i^{Neu})_{1\le i\le N}$ (see eq.~\eqref{eq:AiNeu}) and $(\tilde C_i)_{1\le i\le N}$ (see eq.~\eqref{eq:Ctildei}) are easily obtained by restricting the corresponding variational formulations to adequate local subdomains. Note that these matrices are different from the restriction of the global matrices $A$ and $C$ to the local degrees of freedom. The domain decomposition algorithm presented in this paper is implemented on top of the \textit{ffddm} framework, a set of parallel \textit{FreeFEM} scripts implementing Schwarz domain decomposition methods. \textit{ffddm} already implements the GenEO method~\cite{Spillane:2014:ASC} for SPD problems, and its building blocks are designed to simplify the implementation and prototyping of new domain decomposition methods such as the saddle point solver presented in this paper. The \href{https://doc.freefem.org/documentation/ffddm/index.html}{\textit{ffddm} documentation} is available on the FreeFem.org web page, see~\cite{FFD:Tournier:2020}.\\

%\textcolor{red}{In the following numerical experiments, the iteration counts assess the control of the condition number via the adaptive coarse spaces. We also report timings along with iteration counts. One thing to keep in mind is that our first implementation is not optimal and can be improved in various ways, in particular clear bottlenecks are not handled optimally and we get in more details below when we discuss timings. Nonetheless, we think it is interesting to include timings even at this stage, to help identify computational bottlenecks and understand how to move forward with the implementation.}\\

Numerical results are obtained on the french GENCI supercomputer \textit{Occigen}, on the Haswell partition composed of 50544 cores of Intel Xeon E5-2690V3 processors clocked at \SI{2.6}{\giga\hertz}. The interconnect is an InfiniBand FDR 14 pruned fat tree. We use Intel compilers, the Intel Math Kernel Library and Intel MPI.
%Binaries and shared libraries are compiled with Intel compilers and Math Kernel Library in their version 2019.4.243. The MPI implementation is Intel MPI 2018.1.163.\\

As is usually done in domain decomposition methods, we assign one subdomain per MPI process. Our implementation is pure MPI and no multithreading is done ; we assign one MPI process per computing core. The mesh of the computational domain is partitioned using the automatic graph partitioner \textit{Metis}~\cite{METIS} (see Figure~\ref{fig:beam}). Local subdomain matrices are factorized by the sparse direct solver \textit{MUMPS}~\cite{Amestoy:2001:FAM}. Local eigenvalue problems are solved with \textit{Arpack}~\cite{lehoucq:1998:arpack} ; both libraries are interfaced with \textit{FreeFEM}. GenEO coarse space matrices $R_0 A R_0^T$ in~\eqref{eq:M2} and $Z_{S_1}^T S_1 Z_{S_1}$ in~\eqref{eq:P0tilde} are assembled and factorized in a distributed manner on a few cores (24 in most of the experiments) using the parallel solver \textit{MUMPS}.

\begin{figure}[h!]
\centering
\includegraphics[width=0.49\textwidth]{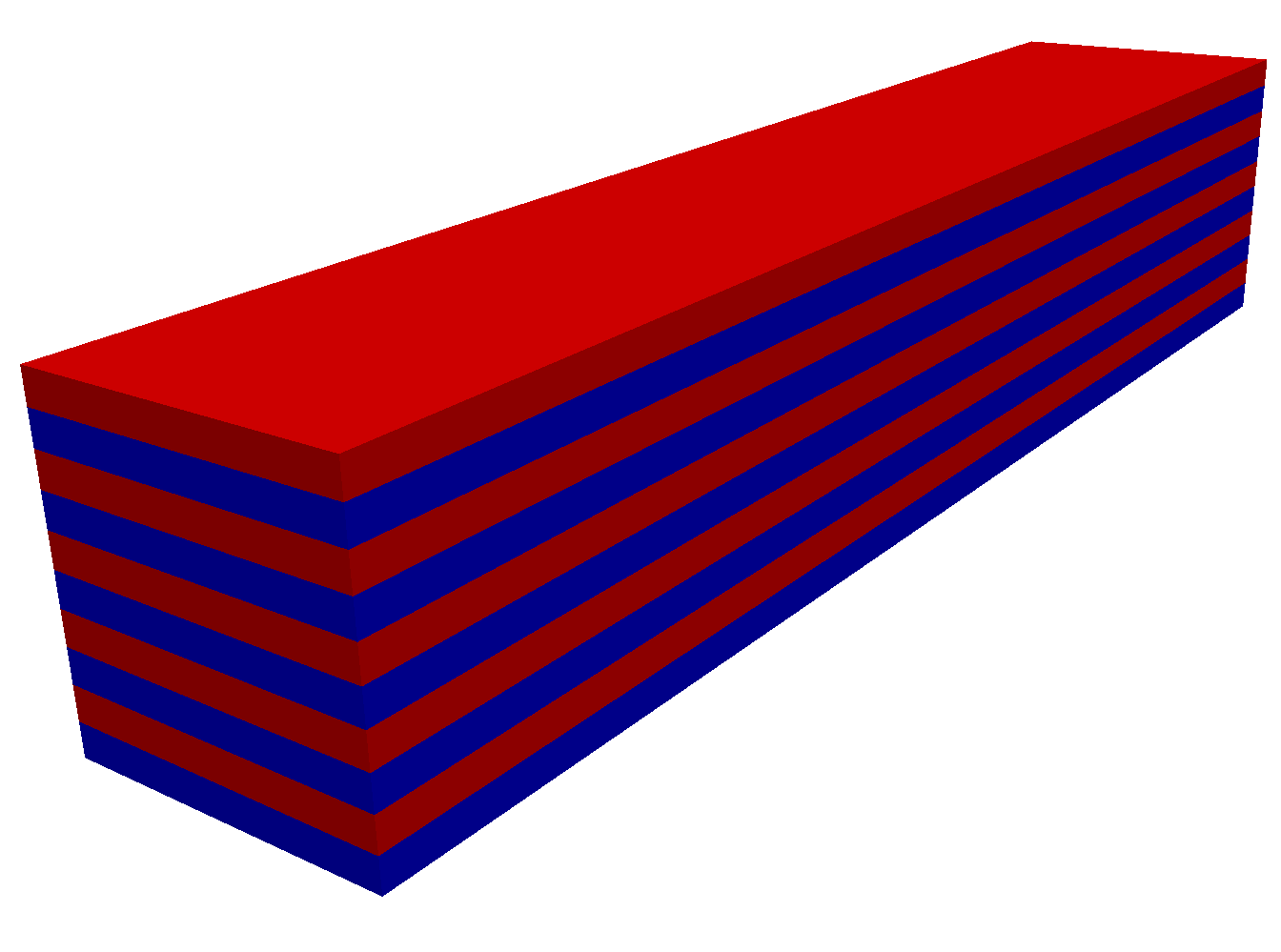}
\includegraphics[width=0.49\textwidth]{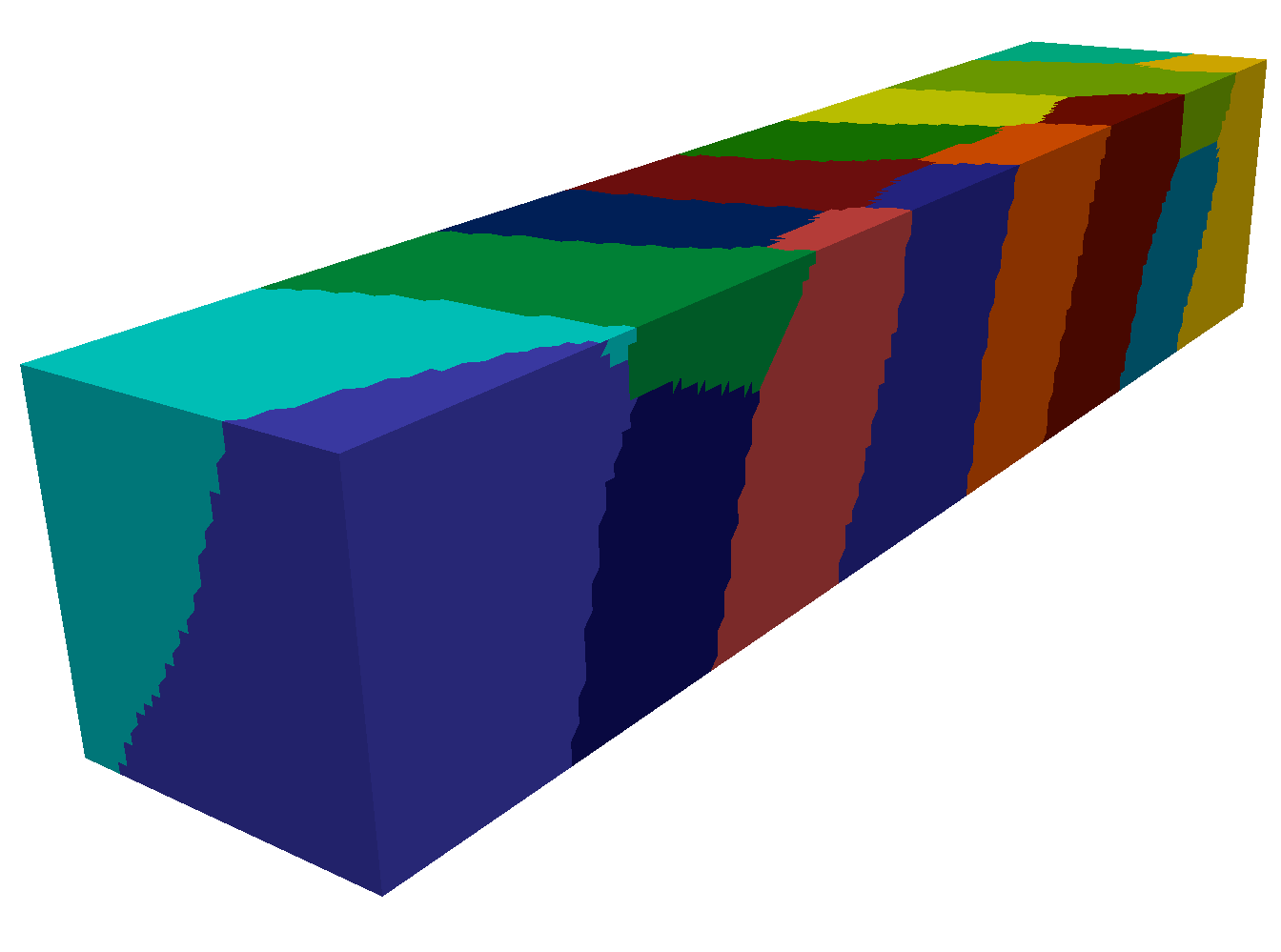}
\caption{Heterogeneous beam composed of 10 alternating layers of rubber and steel. Coefficient distribution (left) and mesh partitioning into 16 subdomains by the automatic graph partitioner \textit{Metis} (right).}
\label{fig:beam}
\end{figure}

For illustration purposes, we represent in Figure~\ref{num:ev} (top) the eigenvalues of the local GenEO eigenvalue problems for both coarse spaces, $V_0$ for $A$ and $\tilde W_0$ for $S_1$ (the former corresponding to eq.~\eqref{eq:geneoA} and the latter to eq.~\eqref{eq:geneoforS1}), for the heterogeneous beam problem with 16 million degrees of freedom, corresponding to the first row of Table~\ref{num:newtable}. The figures show the inverse of the first 40 largest eigenvalues for 10 of the subdomains for the experiment corresponding to the first row of Table~\ref{num:newtable}, so that the eigenvectors corresponding to the smallest values on the graphs (below the dashed line) will be selected to enter the coarse space.  For comparison, we also solve the constant coefficient problem corresponding to an homogeneous steel (and compressible) beam and show the eigenvalues in Figure~\ref{num:ev} (bottom).

\begin{figure}[h!]
\centering
\includegraphics[width=0.49\textwidth,trim= 2cm 2cm 1cm 2cm]{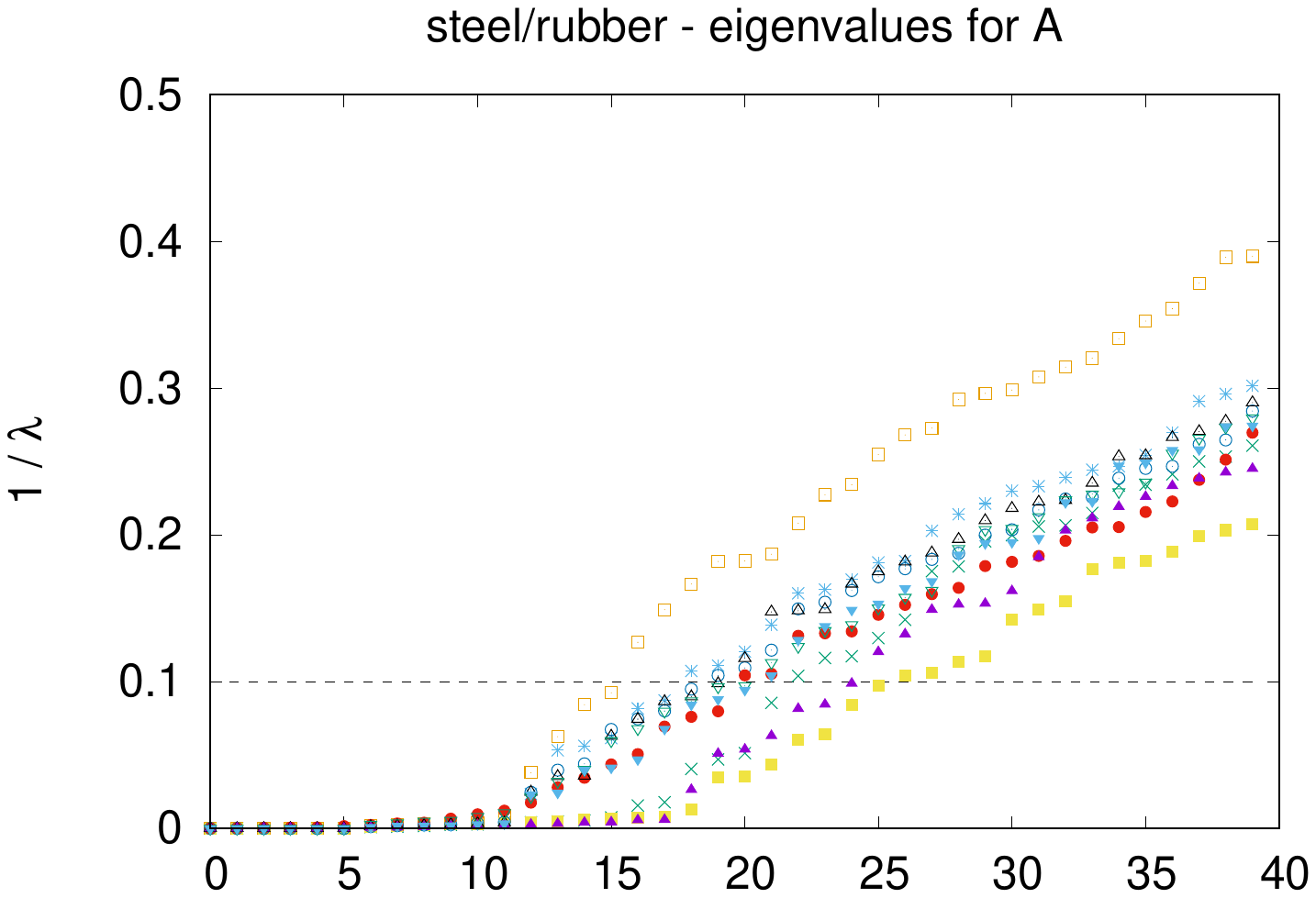}
\includegraphics[width=0.49\textwidth,trim= 2cm 2cm 1cm 2cm]{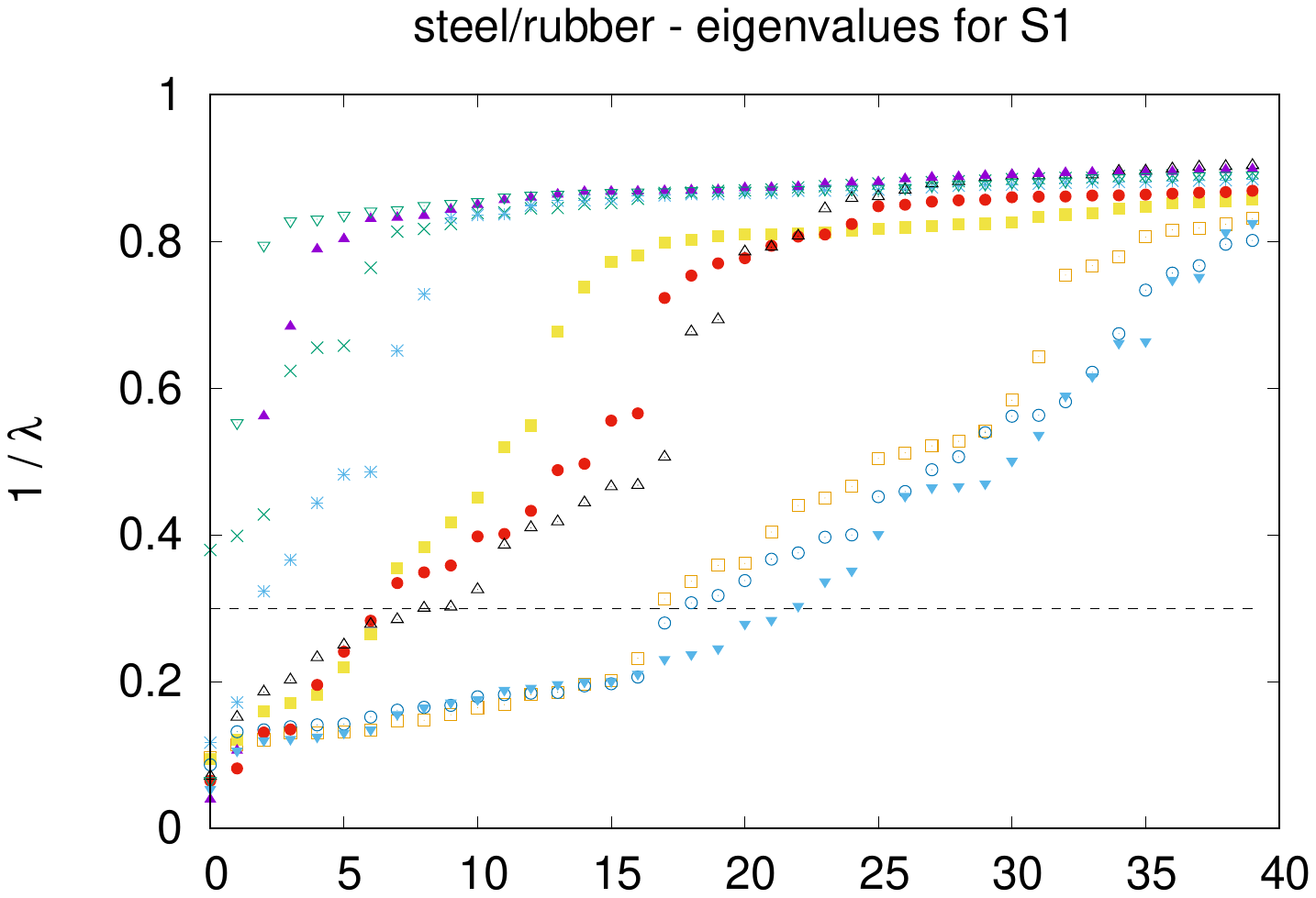}
\includegraphics[width=0.49\textwidth,trim= 2cm 2cm 1cm 2cm]{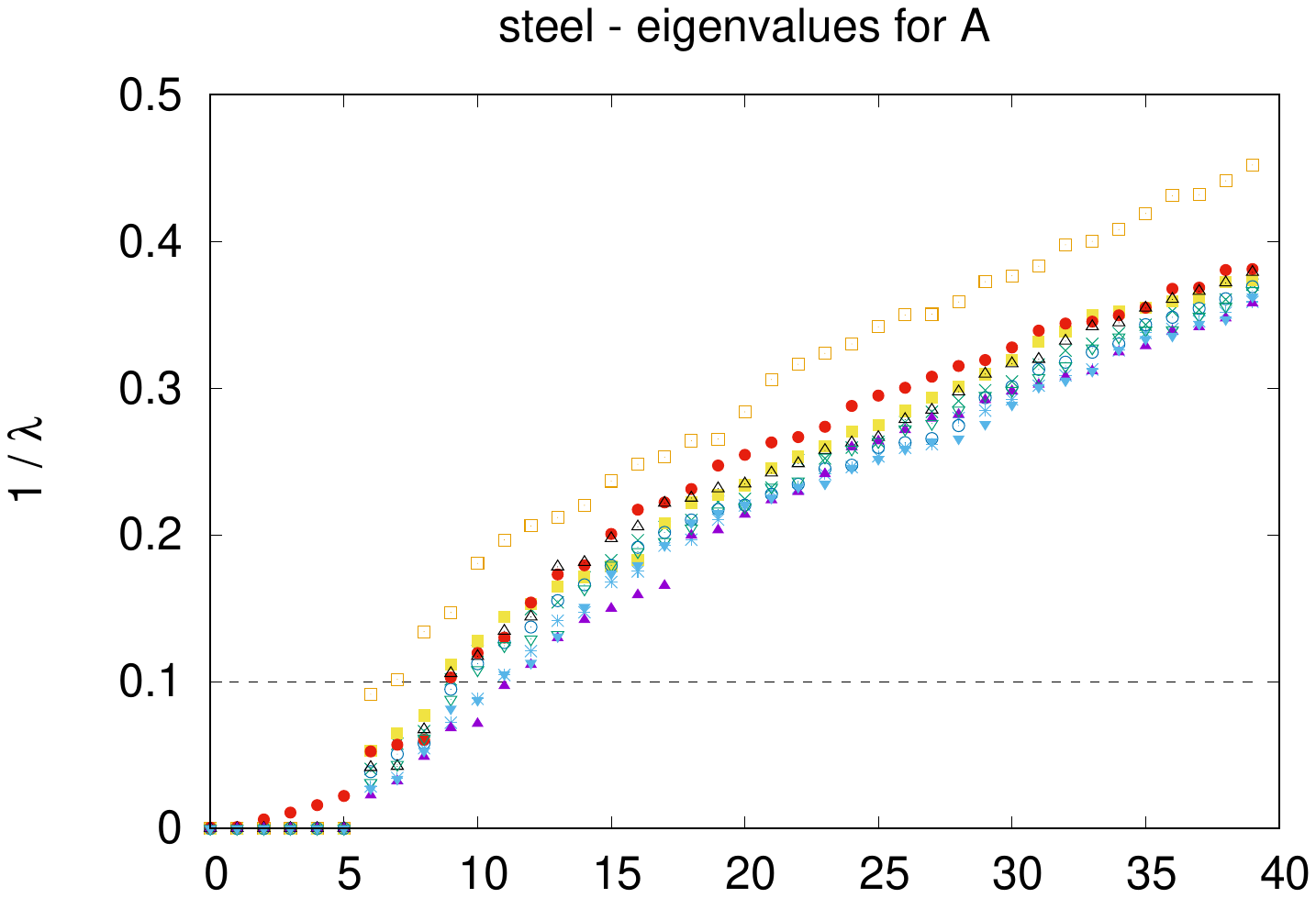}
\includegraphics[width=0.49\textwidth,trim= 2cm 2cm 1cm 2cm]{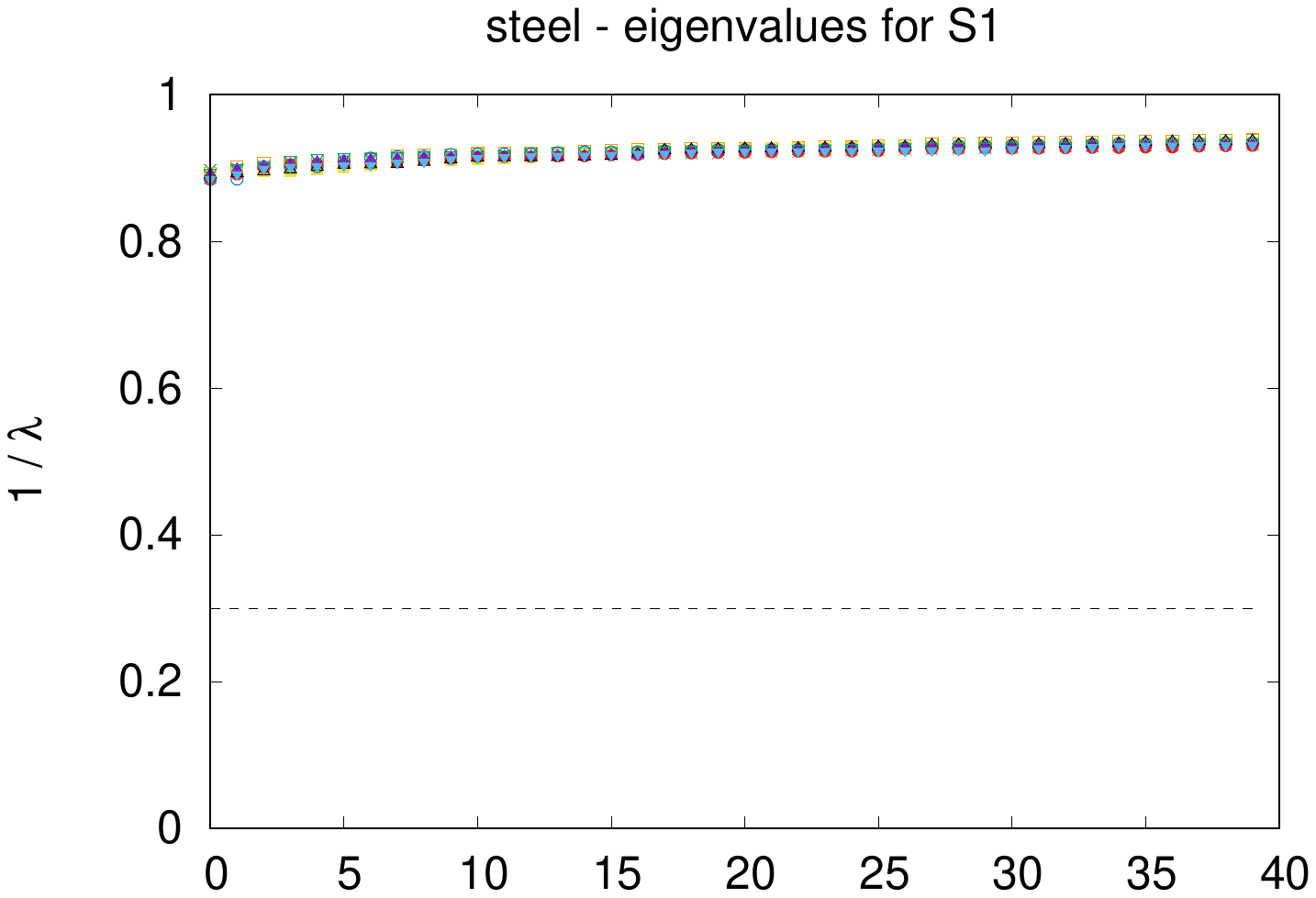}
\caption{Top: Heterogeneous steel and rubber beam. Bottom: Homogeneous steel beam. Inverse of the eigenvalues of the local GenEO eigenvalue problems for both coarse spaces, $V_0$ for $A$ (left) and $\tilde W_0$ for $S_1$ (right), for 10 of the subdomains.}% for the experiment corresponding to the first row of Table~\ref{num:newtable}.}
\label{num:ev}
\end{figure}

%\begin{figure}[h!]
%\centering
%\includegraphics[width=0.49\textwidth,trim= 2cm 2cm 1cm 2cm]{newev_A_acier.pdf}
%\includegraphics[width=0.49\textwidth,trim= 2cm 2cm 1cm 2cm]{newev_S1_acier.pdf}
%\caption{Homogeneous steel beam: inverse of the eigenvalues of the local GenEO eigenvalue problems for both coarse spaces, $V_0$ for $A$ (left) and $\tilde W_0$ for $S_1$ (right), for 10 of the subdomains.}
%\label{num:ev_steel}
%\end{figure}

We can see the effect of heterogeneities on the spread of eigenvalues for different subdomains compared to the homogeneous case. In addition, we retrieve the 6 eigenvalues corresponding to the rigid body modes for $A$ in the homogeneous case, and we see that we need a larger set of eigenvectors in order to build a robust coarse space in the heterogeneous case. Figure~\ref{num:ev} also shows that there is no need for a coarse space for $S_1$ in the compressible homogeneous case. A strong feature of the GenEO method is that relevant eigenvectors to enter the coarse space are selected automatically, adapting to the problem at hand and its spatial heterogeneity. Moreover, the robustness of the model does not rely on a specific partitioning, which allows the use of automatic graph partitioners such as \textit{Metis}~\cite{METIS} or \textit{Scotch}~\cite{webscotch}.

\subsection{Parameters of the method}

The method has a few parameters in play:

\begin{itemize}
\item The number of layers of mesh elements in the overlap region between subdomains is 2 for the velocity blocks (corresponding to $R_i$ in~\eqref{eq:M2}) and 4 for the pressure blocks (corresponding to $\tilde R_i$ in~\eqref{eq:MS}). This corresponds to the minimum overlap that satisfies relation~\eqref{eq:Ritilde} for a symmetric construction of the overlap between subdomains.
\item For the heterogeneous beam problem, we set the threshold $\tau_A$ for selecting the local eigenvectors entering the coarse space $V_0$ to 10 (corresponding to $0.1$ on Figure~\ref{num:ev}, left). The threshold  $\tau_{S_1}$ for selecting the local eigenvectors entering the coarse space $\tilde W_0$ is set to $3.33$ (corresponding to $0.3$ on Figure~\ref{num:ev}, right). 

The selection of these thresholds is based on the fact that the iteration counts are very stable in weak and strong scaling experiments. It means that for a given physics the parameters  $\tau_A$ and $\tau_{S_1}$ can be tuned for a small test case and then used in large scale experiments. 

%We selected a somewhat large threshold $\tau_A$ in order to obtain a not too large coarse space $V_0$ and thus a matrix $M_{A_0}$ \eqref{eq:shermannMorrisonS} of reasonable size. Recall that $M_{A_0}$ is dense and constitutes a bottleneck of the method when the number of subdomains grows.

\end{itemize}

\subsection{Weak scalability test for heterogeneous steel and rubber beam}

Here we present weak scalability results for the heterogeneous beam composed of 10 alternating layers of rubber $(E_1,\nu_1)=(\num{1e+7},0.4999)$ and steel $(E_2,\nu_2)=(\num{2e+9},0.35)$. Local problem size is kept roughly constant as $N$ grows, and the total number of dofs $n$ goes from $16$ million on $262$ cores to $1$ billion on $16800$ cores. We use the original definition of $M_{S_1}^{-1}$~\eqref{eq:S1prec}.%in order to simplify the communication pattern and computations.

We report in Table~\ref{num:newtable} the iteration counts and computing times for the DD saddle point solver Algorithm~\ref{alg:ddsaddlepointsolver}. Note that in Algorithm~\ref{alg:ddsaddlepointsolver} we replace PCG by right-preconditioned GMRES for step 3. The stopping criterion is a tolerance smaller than $10^{-5}$. Moreover, we use flexible GMRES, as we solve~\eqref{eq:newinner} inexactly using GMRES with a tolerance of $10^{-2}$ in order to apply $N_S^{-1}$.

In order, columns correspond to: number of cores, number of dofs $n$, size of the coarse space for A $dim(V_0)$, size of the coarse space for $S_1$ $dim(\tilde W_0)$, setup time corresponding to the assembly and factorization of the various local and coarse operators, number of outer GMRES iterations, GMRES computing time, total computing time (setup + GMRES) and average number of inner GMRES iterations for each solution of~\eqref{eq:newinner}. All timings are reported in seconds.

\begin{table}[h!]
\small
%\textcolor{red}{new:}
\newtablefromcsv{}{nolump-newweak.csv}
%\textcolor{red}{old:}
%\newtablefromcsv{}{newweak.csv}
\caption{Weak scaling experiment for 3D heterogeneous elasticity: beam with 10 alternating layers of steel and rubber. Reported iteration counts and timings for DD saddle point algorithm~\ref{alg:ddsaddlepointsolver}.}
\label{num:newtable}
\end{table}

\paragraph{Iteration counts.} We first discuss iteration counts. We see that outer iteration count remains stable, between 21 and 32. The inner iteration count is also stable and remains around 11. We also observed (figures are not reported here) than the inner GMRES tolerance of $10^{-2}$ does not affect the outer iteration count compared to an accurate solution with a stricter tolerance of $10^{-5}$, and allows a significant reduction in inner iteration count. For example, 11 iterations on average instead of 28 on 1050 cores for the same outer iteration count of 22, leading to a decrease from 1178.2 to 665.9 seconds in GMRES timing. % 9 / 28

\paragraph{Timings.} In terms of setup timings, the computing time remains relatively stable, with roughly $15\%$ increase for a factor of 64 in problem size. Around $60\%$ of the setup time is spent in the solution of the  eigenvalue problems~\eqref{eq:geneoforS1} for $S_1$.
%The increase in setup time at 8400 and 16800 cores can be explained by the larger factorization cost of the coarse problems ; for example, assembling and factorizing $R_0 A R_0^T$ and $Z_{S_1}^T S_1 Z_{S_1}$ takes 90

The solution time stays relatively stable up to 4200 cores, where it starts to degrade. This can be related to the increased cost of the coarse space solves with matrices $R_0 A R_0^T$ and $Z_{S_1}^T S_1 Z_{S_1}$ as their size increases: total time spent in coarse space solves is 14.7, 62.1 and 679.3 seconds on 262, 4200 and 16800 cores respectively. A possible improvement would be to use a multi-level method to solve the coarse problems iteratively.

%The solution time increases from $191$ to $458$ seconds for a problem $16$ times the size. This increase can be explained by the slight increase in inner and outer GMRES iterations, as well as the fact that coarse problems $R_0 A R_0^T$ and $Z_{S_1}^T S_1 Z_{S_1}$ become increasingly more expensive to solve as their size increases. Additionally, solving for $A^{-1}$ exactly in step 3 of Algorithm~\ref{alg:ddsaddlepointsolver} at each iteration, which is done in practice using GMRES with preconditioner $M_A^{-1}$ and tolerance $10^{-6}$, amounts to around $60\%$ of the solution time, from $116$s on 512  cores to $284$s on 8192 cores.\\ %the coarse problem for $S_1$ becomes the main bottleneck as $N$ increases. We recall that the pattern of $Z_{S_1}^T S_1 Z_{S_1}$ is denser than that of $R_0 A R_0^T$, as $S_1$ is the sum of local dense Schur complements {\bf FAIRE SENTIR QUE C'EST EFFECTIVEMENT LE CAS?}.

\subsection{Strong scalability test for heterogeneous steel and rubber beam}

Strong scalability results for the heterogeneous beam composed of 10 alternating layers of rubber and steel are presented in Table~\ref{num:newtablestrong}. The problem size is $27.5$ million and the strong scaling test ranges from 525 to 4200 cores. Iteration counts and computing times for the DD saddle point solver Algorithm~\ref{alg:ddsaddlepointsolver} are reported.

\begin{table}[h!]
\small
%\textcolor{red}{new:}
\newtablefromcsv{}{nolump-newstrong.csv}
%\textcolor{red}{old:}
%\newtablefromcsv{}{newstrong.csv}
\caption{Strong scaling experiment for 3D heterogeneous elasticity: beam with 10 alternating layers of steel and rubber. Reported iteration counts and timings for DD saddle point Algorithm~\ref{alg:ddsaddlepointsolver}.}
\label{num:newtablestrong}
\end{table}

\paragraph{Iteration counts.} Outer iteration count remains stable, with a slight increase from 21 to 23. Inner iteration counts remains also stable, even slightly decreasing from 12 to 9.

\paragraph{Timings.} We see that the setup timing decreases accordingly as the subdomains shrink in size, from $526.6$ seconds on 525 cores to $103.8$ seconds on 4200 cores; the speedup efficiency with respect to 525 cores ranges from $99\%$ on 1050 cores to $63\%$ on 4200 cores. % (the superlinear speedup is explained by the superlinear behaviour of the direct solver for local subdomain problems as their size decreases)

We see a similar trend for the solution time, ranging from $519.5$ seconds on 525 cores to $91.3$ seconds on 4200 cores ; the speedup efficiency with respect to 525 cores ranges from $116\%$ on 1050 cores to $71\%$ on 4200 cores. This decrease in efficiency can be explained by the increased relative cost of coarse space solves as subdomains get smaller: from $3\%$ on 525 cores to $30\%$ on 4200 cores. The added overlap also plays a greater role in the loss of efficiency as subdomains get smaller.\\

\begin{table}[ht!]
	\begin{center}
	\begin{tabular}{|S[table-format=7.0]|S[table-format=4.0]|S[table-format=4.1]|S[table-format=1.1]|S[table-format=4.1]|S[table-format=3.1]|S[table-format=2.0]|S[table-format=3.1]|S[table-format=3.1]|}
	\cline{3-9}
	\multicolumn{2}{c}{} & \multicolumn{3}{|c|}{MUMPS} & \multicolumn{4}{c|}{DD saddle point solver} \\ \hline
	{$n$} & {\#cores} & {setup(s)} & {solve(s)} & {total(s)} & {setup(s)} & {\#It} & {gmres(s)} & {total(s)} \\ \hline
	 139809 &    16 &     7.1  &     0.1  &      7.2 &     27.1 &   18 &    19.7  &    46.8  \\ \hline
	 1058312 &    32 &    85.7  &     0.8  &     86.5 &    166.2 &   20 &   137.2  &   303.4  \\
	 1058312 &    65 &    71.0  &     0.6  &     71.6 &     91.0 &   21 &    77.1  &   168.1  \\
	 1058312 &   131 &    63.2  &     0.5  &     63.7 &     59.7 &   24 &    49.7  &   109.4  \\ \hline
	 3505582 &    55 &   477.8  &     3.7  &    481.5 &    404.1 &   24 &   430.1  &   834.2  \\
	 3505582 &   110 &   392.3  &     2.3  &    394.6 &    242.5 &   23 &   212.8  &   455.3  \\
	 3505582 &   221 &   387.0  &     2.1  &    389.1 &    134.8 &   23 &   109.4  &   244.2  \\
	 3505582 &   442 &   453.9  &     2.2  &    456.1 &     88.2 &   24 &    68.6  &   156.8  \\ \hline
	8235197 &   262 &     OOM  &       /  &        / &    278.5 &   25 &   264.3  &   542.8  \\
	8235197 &   525 &  1622.1  &     6.1  &   1628.2 &    172.1 &   24 &   136.0  &   308.1  \\
	8235197 &  1050 &  1994.3  &     7.4  &   2001.7 &    136.5 &   25 &    99.7  &   236.2  \\ \hline
	\end{tabular}
	\end{center}
	\caption{Comparison with the parallel sparse direct solver \textit{MUMPS} for 3D heterogeneous elasticity: beam with 10 alternating layers of steel and rubber. Reported timings for four discretization levels while also varying the number of cores (OOM means the computation ran out of available memory).}
	\label{num:cmpmumps}
\end{table}

In Table~\ref{num:cmpmumps}, we compare the performance of the solver to the parallel sparse direct solver \textit{MUMPS} for the heterogeneous steel and rubber beam test case with four discretization levels, while also varying the number of cores. As we can see, \textit{MUMPS} is comparatively more efficient for smaller problems, with for example a total time of 86.5 seconds compared to 303.4 seconds for our saddle point solver for 1 million unknowns on 32 cores. However, as expected, we see a large increase in memory and computational cost as the size of the system gets larger: for 8.2 million unknowns, \textit{MUMPS} runs out of memory on 262 cores and solves the problem in 1628.2 seconds on 525 cores, compared to 308.1 seconds for the DD solver. Moreover, we can see from Table~\ref{num:cmpmumps} that the DD saddle point solver offers much better strong scalability.\\

In Figure~\ref{num:gmres}, we plot the convergence history of GMRES for both compressible homogeneous and heterogeneous steel-rubber cases with $27.5$ million unknowns on 525 cores (left), with a comparison to the standard one-level Additive Schwarz Method (ASM) from PETSc~\cite{petsc-efficient} on the global problem (right), illustrating the difficulty of the test case at hand. Our saddle point solver needs 15 and 21 iterations to converge for the homogeneous and heterogeneous cases respectively, compared to 856 and 2880 for the one-level method, with significant undesired plateaux. For the same physical test case, iteration counts are better that in \cite{haferssas:2017:additive} but timings are not as good. As mentioned in the introduction, our method has the advantage of a provable convergence estimate and to not depend on the design of specific absorbing conditions for the elasticity system.\\

\begin{figure}[h!]
\centering
\includegraphics[height=6.3cm,trim= 6cm 2cm 5.5cm 2cm]{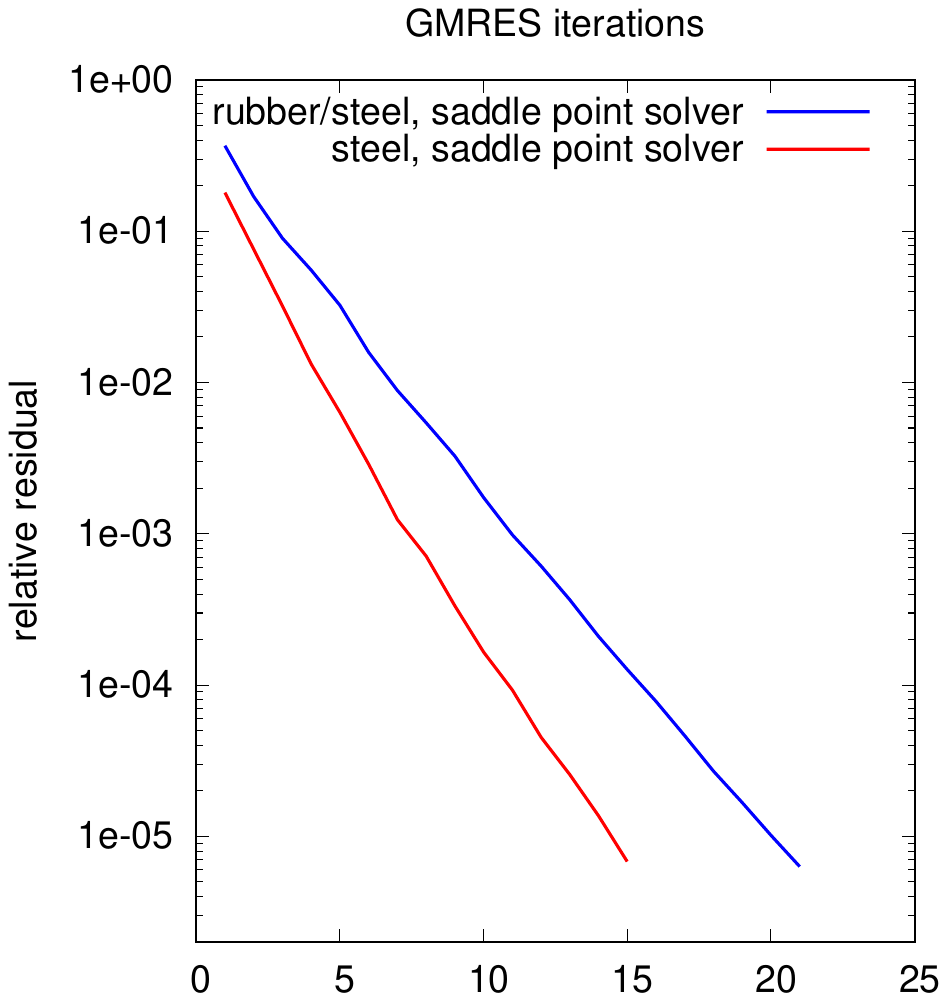}
\includegraphics[height=6.3cm,trim= 2cm 2cm 1cm 2cm]{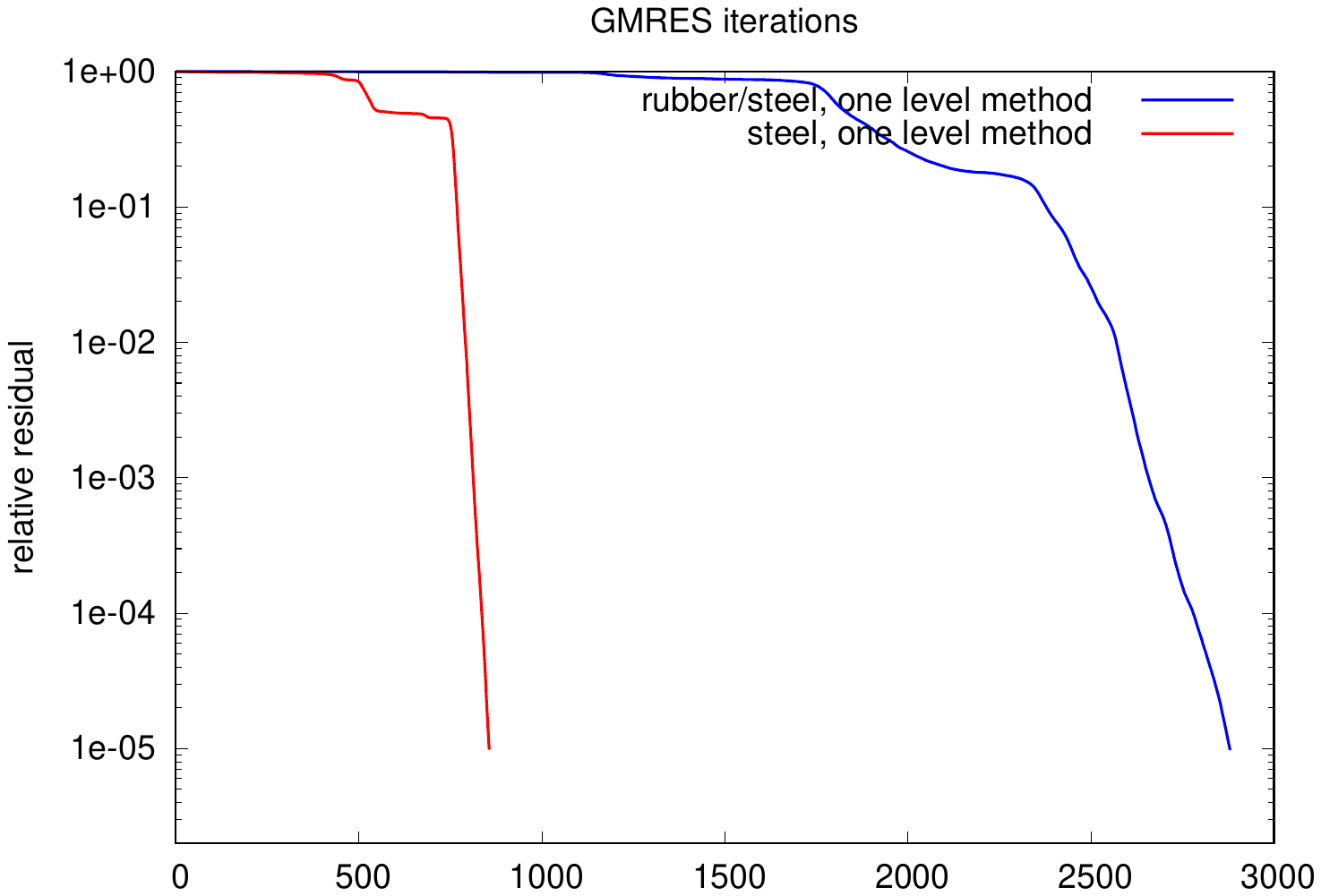}
\caption{GMRES convergence history of the saddle point solver (left) compared to the one-level Additive Schwarz Method (right) for the homogeneous steel beam and heterogeneous rubber/steel beam problems discretized with $27.5$ million unknowns (corresponding to the first row of Table~\ref{num:newtablestrong}), on 525 cores.}
\label{num:gmres}
\end{figure}

We also performed comparisons with the Geometric Algebraic Multigrid (GAMG) preconditioner from PETSc. We were not able to find a suitable tuning of parameters for GAMG for the saddle point formulation. However, we performed comparisons between GAMG and standard GenEO for the velocity formulation on the homogeneous beam, varying the Poisson ratio $\nu$ from 0.48 to 0.499. The GenEO threshold $\tau$ is set to $3.33$, and we select at most 80 eigenvectors in each subdomain. Even though GAMG is faster for $\nu \leq 0.49$, we can see that GenEO is more robust as $\nu$ increases. In particular, GAMG fails to converge in 2000 iterations for $\nu \geq 0.495$.

\begin{table}[ht!]
	\begin{center}
	\begin{tabular}{|S[table-format=1.3]|S[table-format=5.0]|S[table-format=2.1]|S[table-format=5.0]|S[table-format=2.1]|S[table-format=2.0]|S[table-format=2.1]|S[table-format=2.1]|}
	\cline{2-8}
	\multicolumn{1}{c}{131 cores} & \multicolumn{2}{|c|}{GAMG} & \multicolumn{5}{c|}{DD saddle point solver} \\ \hline
	{$\nu$} & {\#It} & {total(s)} & {$dim(V_0)$} & {setup(s)} & {\#It} & {gmres(s)} & {total(s)} \\ \hline
	0.48 & 60 & 67.1 & 10480 & 200.7 & 24 & 11.2 & 212.0 \\ \hline
	0.485 & 109 & 89.0 & 10480 & 199.5 & 27 & 12.7 & 212.2 \\ \hline
	0.49 & 210 & 137.0 & 10480 & 202.0 & 32 & 15.0 & 217.0 \\ \hline
	0.495 & {>2000} & {/} & 10480 & 199.9 & 43 & 20.2 & 220.1 \\ \hline
	0.499 & {>2000} & {/} & 10480 & 199.2 & 99 & 48.6 & 247.7 \\ \hline
	\end{tabular}

	\begin{tabular}{|S[table-format=1.3]|S[table-format=5.0]|S[table-format=2.1]|S[table-format=5.0]|S[table-format=2.1]|S[table-format=2.0]|S[table-format=2.1]|S[table-format=2.1]|}
	\cline{2-8}
	\multicolumn{1}{c}{525 cores} & \multicolumn{2}{|c|}{GAMG} & \multicolumn{5}{c|}{DD saddle point solver} \\ \hline
	{$\nu$} & {\#It} & {total(s)} & {$dim(V_0)$} & {setup(s)} & {\#It} & {gmres(s)} & {total(s)} \\ \hline
	0.48 & 56 & 25.5 & 41766 & 60.4 & 18 & 5.0 & 65.4 \\ \hline
	0.485 & 60 & 26.1 & 41984 & 60.9 & 20 & 5.3 & 66.2 \\ \hline
	0.49 & 116 & 33.3 & 42000 & 60.4 & 23 & 5.9 & 66.3 \\ \hline
	0.495 & {>2000} & {/} & 42000 & 60.4 & 32 & 7.6 & 68.1 \\ \hline
	0.499 & {>2000} & {/} & 42000 & 60.6 & 95 & 20.3 & 81.0 \\ \hline
	\end{tabular}
	\end{center}
	\caption{GAMG versus standard GenEO for the velocity formulation on the homogeneous beam discretized with 7.9 million unknowns, using 131 and 525 cores. Reported iteration counts and timings for different values of the Poisson ratio $\nu$ ranging from 0.48 to 0.499.}
	\label{num:cmpgamg}
\end{table}

\section{Conclusion and outlook} % (fold)
\label{sec:conclusion}

Under the assumption that the diagonal block matrices of a saddle point problem  are spectrally equivalent to a sum of positive semi definite matrices, we have introduced an adaptive domain decomposition (DD) method. For this, two coarse spaces are built by solving generalized eigenvalue problems, one for the primal unknowns and the second one for the dual unknowns. The robustness of the method was assessed on a notoriously difficult three dimensional elasticity problem for a steel-rubber structure discretized  with continuous pressure.

Several issues deserve further investigations. First a multilevel method with more than two levels would enable even larger and possibly faster simulations. Also, the tests were performed with FreeFem scripts using the standalone \textit{ffddm}~\cite{FFD:Tournier:2020} framework. The integration of the method in the C++/MPI library \textit{hpddm}~\cite{Jolivet:2014:HPD} could lead to faster codes and a more general diffusion of the saddle point preconditioner. In a different setting, the design of adaptive coarse space is strongly connected to multiscale finite element (MFE) methods (see~e.g., \cite{efendiev2013generalized,tchelepi2005adaptive,Scheichl2021novel} and references therein) and this work could be used in designing MFE methods for saddle point problem.    

% section conclusion (end)

\section*{Acknowledgment}

This work was granted access to the HPC resources of OCCIGEN@CINES under the allocation  2020-067730 granted by GENCI.

\section*{Code reproducibility}
Our numerical results can be reproduced running the script \url{https://github.com/FreeFem/FreeFem-sources/blob/master/examples/ffddm/elasticity_saddlepoint.edp} available in the FreeFem distribution starting from version $4.10$.

	\bibliographystyle{plain}
	\bibliography{../../elasticite/NotesdeCours/Poly/bookddm}

\begin{thebibliography}{10}

\bibitem{Amestoy:2001:FAM}
Patrick~R. Amestoy, Iain~S. Duff, Jean-Yves L'Excellent, and Jacko Koster.
\newblock A fully asynchronous multifrontal solver using distributed dynamic
  scheduling.
\newblock {\em SIAM J. Matrix Analysis and Applications}, 23(1):15--41, 2001.

\bibitem{Arnold:2000:MHH}
Douglas~N. Arnold, Richard~S. Falk, and Ragnar Winther.
\newblock Multigrid in h (div) and h (curl).
\newblock {\em Numerische Mathematik}, 85(2):197--217, Apr 2000.

\bibitem{petsc-efficient}
Satish Balay, William~D. Gropp, Lois~Curfman McInnes, and Barry~F. Smith.
\newblock Efficient management of parallelism in object oriented numerical
  software libraries.
\newblock In E.~Arge, A.~M. Bruaset, and H.~P. Langtangen, editors, {\em Modern
  Software Tools in Scientific Computing}, pages 163--202. Birkh{\"{a}}user
  Press, 1997.

\bibitem{Benzi:2005:NSS}
Michele Benzi, Gene~H. Golub, and J{\"o}rg Liesen.
\newblock Numerical solution of saddle point problems.
\newblock {\em Acta Numer.}, 14:1--137, 2005.

\bibitem{Benzi:2008:SPT}
Michele Benzi and Andrew~J. Wathen.
\newblock Some preconditioning techniques for saddle point problems.
\newblock In {\em Model order reduction: theory, research aspects and
  applications}, volume~13 of {\em Math. Ind.}, pages 195--211. Springer,
  Berlin, 2008.

\bibitem{Brezzi:2012:MHF}
Franco Brezzi and Michel Fortin.
\newblock {\em Mixed and hybrid finite element methods}, volume~15.
\newblock Springer Science \& Business Media, 2012.

\bibitem{Cahouet:1988:F3F}
J~Cahouet and J-P Chabard.
\newblock Some fast 3d finite element solvers for the generalized stokes
  problem.
\newblock {\em International Journal for Numerical Methods in Fluids},
  8(8):869--895, 1988.

\bibitem{webscotch}
C.~Chevalier and F.~Pellegrini.
\newblock {PT-SCOTCH}: a tool for efficient parallel graph ordering.
\newblock {\em Parallel Computing}, {6-8}(34):318--331, 2008.

\bibitem{sturler:2005:BDC}
Eric de~Sturler and J\"{o}rg Liesen.
\newblock Block-diagonal and constraint preconditioners for nonsymmetric
  indefinite linear systems. {I}. {T}heory.
\newblock {\em SIAM J. Sci. Comput.}, 26(5):1598--1619, 2005.

\bibitem{Dolean:2015:IDDSiam}
Victorita Dolean, Pierre Jolivet, and Fr\'ed\'eric Nataf.
\newblock {\em An Introduction to Domain Decomposition Methods: algorithms,
  theory and parallel implementation}.
\newblock SIAM, 2015.

\bibitem{Drzisga:2018:analysis}
Daniel Drzisga, Lorenz John, Ulrich Rude, Barbara Wohlmuth, and Walter
  Zulehner.
\newblock On the analysis of block smoothers for saddle point problems.
\newblock {\em SIAM Journal on Matrix Analysis and Applications},
  39(2):932--960, 2018.

\bibitem{efendiev2013generalized}
Yalchin Efendiev, Juan Galvis, and Thomas~Y Hou.
\newblock Generalized multiscale finite element methods (gmsfem).
\newblock {\em Journal of Computational Physics}, 251:116--135, 2013.

\bibitem{Farhat:1991:MFE}
Charbel Farhat and Francois-Xavier Roux.
\newblock A method of {F}inite {E}lement {T}earing and {I}nterconnecting and
  its parallel solution algorithm.
\newblock {\em Int.\ J.\ Numer.\ Meth.\ Engrg.}, 32:1205--1227, 1991.

\bibitem{Farrell:2019:ALP}
Patrick~E. Farrell, Lawrence Mitchell, and Florian Wechsung.
\newblock An {A}ugmented {L}agrangian {P}reconditioner for the 3{D}
  {S}tationary {I}ncompressible {N}avier--{S}tokes {E}quations at {H}igh
  {R}eynolds {N}umber.
\newblock {\em SIAM J. Sci. Comput.}, 41(5):A3073--A3096, 2019.

\bibitem{Griebel:1995:ATA}
M.~Griebel and P.~Oswald.
\newblock On the abstract theory of additive and multiplicative {S}chwarz
  algorithms.
\newblock {\em Numer. Math.}, 70(2):163--180, 1995.

\bibitem{haferssas:2017:additive}
Ryadh Haferssas, Pierre Jolivet, and Fr{\'e}d{\'e}ric Nataf.
\newblock An additive schwarz method type theory for lions's algorithm and a
  symmetrized optimized restricted additive schwarz method.
\newblock {\em SIAM Journal on Scientific Computing}, 39(4):A1345--A1365, 2017.

\bibitem{Hecht:2012:NDF}
F.~Hecht.
\newblock New development in {F}reefem++.
\newblock {\em J. Numer. Math.}, 20(3-4):251--265, 2012.

\bibitem{Hiptmair:1998:MMM}
R.~Hiptmair.
\newblock Multigrid method for {M}axwell's equations.
\newblock {\em SIAM J. Numer. Anal.}, 36(1):204--225, 1998.

\bibitem{Hiptmair:1997:MMH}
Ralf Hiptmair.
\newblock Multigrid method for h (div) in three dimensions.
\newblock {\em Electron. Trans. Numer. Anal}, 6(1):133--152, 1997.

\bibitem{tchelepi2005adaptive}
Patrick Jenny, Seong~H Lee, and Hamdi~A Tchelepi.
\newblock Adaptive multiscale finite-volume method for multiphase flow and
  transport in porous media.
\newblock {\em Multiscale Modeling \& Simulation}, 3(1):50--64, 2005.

\bibitem{Jolivet:2014:HPD}
Pierre Jolivet and Fr\'ed\'eric Nataf.
\newblock Hpddm: {High-Performance Unified framework for Domain Decomposition
  methods, MPI-C++ library}.
\newblock {https://github.com/hpddm/hpddm}, 2014.

\bibitem{METIS}
G.~Karypis and V.~Kumar.
\newblock {METIS}: {A} software package for partitioning unstructured graphs,
  partitioning meshes, and computing fill-reducing orderings of sparse
  matrices.
\newblock Technical report, Department of Computer Science, University of
  Minnesota, 1998.
\newblock http://glaros.dtc.umn.edu/gkhome/views/metis.

\bibitem{Klawonn:1998:OPC}
Axel Klawonn.
\newblock An optimal preconditioner for a class of saddle point problems with a
  penalty term.
\newblock {\em SIAM Journal on Scientific Computing}, 19(2):540--552, 1998.

\bibitem{lehoucq:1998:arpack}
Richard~B Lehoucq, Danny~C Sorensen, and Chao Yang.
\newblock {\em ARPACK users' guide: solution of large-scale eigenvalue problems
  with implicitly restarted Arnoldi methods}, volume~6.
\newblock SIAM, 1998.

\bibitem{Scheichl2021novel}
Chupeng Ma, Robert Scheichl, and Tim Dodwell.
\newblock Novel design and analysis of generalized fe methods based on locally
  optimal spectral approximations.
\newblock {\em arXiv preprint arXiv:2103.09545}, 2021.

\bibitem{Murphy:2000:NPI}
Malcolm~F. Murphy, Gene~H. Golub, and Andrew~J. Wathen.
\newblock A note on preconditioning for indefinite linear systems.
\newblock {\em SIAM J. Sci. Comput.}, 21(6):1969--1972, 2000.

\bibitem{Nepomnyaschikh:1991:MTT}
Sergey~V. Nepomnyaschikh.
\newblock Mesh theorems of traces, normalizations of function traces and their
  inversions.
\newblock {\em Sov. J. Numer. Anal. Math. Modeling}, 6:1--25, 1991.

\bibitem{Pasciak:2002:OSM}
Joseph~E Pasciak and Jun Zhao.
\newblock Overlapping schwarz methods in h (curl) on polyhedral domains.
\newblock {\em Journal of Numerical Mathematics}, 10(3):221--234, 2002.

\bibitem{Pavarino:2002:BNN}
Luca~F. Pavarino and Olof~B. Widlund.
\newblock Balancing {N}eumann-{N}eumann methods for incompressible {S}tokes
  equations.
\newblock {\em Comm. Pure Appl. Math.}, 55(3):302--335, 2002.

\bibitem{Rees:2020:EBP}
T.~Rees and M.~Wathen.
\newblock An element-based preconditioner for mixed finite element problems.
\newblock {\em SIAM Journal on Scientific Computing}, 2020.

\bibitem{Reitzinger:2002:AMM}
S.~Reitzinger and J.~Sch{\"o}berl.
\newblock An algebraic multigrid method for finite element discretizations with
  edge elements.
\newblock {\em Numerical Linear Algebra with Applications}, 9(3):223--238,
  2002.

\bibitem{Spillane:2014:ASC}
Nicole Spillane, Victorita Dolean, Patrice Hauret, Fr\'ed\'eric Nataf, Clemens
  Pechstein, and Robert Scheichl.
\newblock Abstract robust coarse spaces for systems of {PDE}s via generalized
  eigenproblems in the overlaps.
\newblock {\em Numer. Math.}, 126(4):741--770, 2014.

\bibitem{Spillane:2013:ASC}
Nicole Spillane and Daniel Rixen.
\newblock Automatic spectral coarse spaces for robust finite element tearing
  and interconnecting and balanced domain decomposition algorithms.
\newblock {\em Internat. J. Numer. Methods Engrg.}, 95(11):953--990, 2013.

\bibitem{Toselli:2005:DDM}
Andrea Toselli and Olof Widlund.
\newblock {\em Domain Decomposition Methods - Algorithms and Theory}, volume~34
  of {\em Springer Series in Computational Mathematics}.
\newblock Springer, 2005.

\bibitem{FFD:Tournier:2020}
Pierre-Henri Tournier and Fr{\'e}d{\'e}ric Nataf.
\newblock {FFDDM}: Freefem domain decomposition methd.
\newblock {https://doc.freefem.org/documentation/ffddm/index.html}, 2019.

\bibitem{Tu:2015:FDT}
Xuemin Tu and Jing Li.
\newblock A {FETI}-{DP} type domain decomposition algorithm for
  three-dimensional incompressible {S}tokes equations.
\newblock {\em SIAM J. Numer. Anal.}, 53(2):720--742, 2015.

\bibitem{Widlund:2021:block}
O~Widlund, Stefano Zampini, S~Scacchi, and L~Pavarino.
\newblock Block feti--dp/bddc preconditioners for mixed isogeometric
  discretizations of three-dimensional almost incompressible elasticity.
\newblock {\em Mathematics of Computation}, 90(330):1773--1797, 2021.

\bibitem{Notay:2019:CSI}
Notay Y.
\newblock Convergence of some iterative methods for symmetric saddle point
  linear systems.
\newblock {\em SIMAX}, 40:122--146, 2019.

\end{thebibliography}

\end{document}